\def\SarielStyle{1}

\documentclass[12pt]{article}%

\ifx\SarielStyle\undefined
\newcommand{\MichaVer}[1]{#1}
\newcommand{\SarielVer}[1]{}
\else
\newcommand{\MichaVer}[1]{}
\newcommand{\SarielVer}[1]{#1}
\fi

\usepackage{latexsym,color,graphicx,amsmath,amssymb}
\usepackage[cm]{fullpage}%
\usepackage{graphicx}%
\usepackage{mleftright}%
\usepackage{xcolor}%
\usepackage{xspace}%
\usepackage{caption}%
\usepackage{animate}%

\usepackage[amsmath,thmmarks]{ntheorem}%
\theoremseparator{.}%

\usepackage{picins}%
\usepackage{paralist}%

\usepackage{titlesec}%
\titlelabel{\thetitle. }%

\titleformat{\subsubsection}[runin]%
  {\normalfont\normalsize\bfseries}{\thesubsubsection.}{0.37em}{}

\titlespacing*{\subsubsection}%
{0pt}{1.1ex plus 0.1ex minus 0.04ex}{1.3ex plus .2ex}

\usepackage{hyperref}%
\hypersetup{%
   breaklinks,%
   colorlinks=true,%
   linkcolor=[rgb]{0.45,0.0,0.0},%
   citecolor=[rgb]{0,0,0.45}
}%

\numberwithin{figure}{section}%
\numberwithin{table}{section}%
\numberwithin{equation}{section}%

   \newtheorem{theorem}{Theorem}[section]%
   \newtheorem{lemma}[theorem]{Lemma}%
   \newtheorem*{restate*}[theorem]{Restatement of }%
   \newtheorem{corollary}[theorem]{Corollary}%
   
   \newcommand{\myqedsymbol}{\rule{2mm}{2mm}}

\theoremstyle{remark}%
\theoremheaderfont{\sf}%
\theorembodyfont{\upshape}%
\newtheorem{defn}[theorem]{Definition}

\theoremheaderfont{\em}%
\theorembodyfont{\upshape}%
\theoremstyle{nonumberplain}%
\theoremseparator{}%
\theoremsymbol{\myqedsymbol}%
\newtheorem{proof}{Proof:}%

\newcommand{\remove}[1]{}

\definecolor{blue25}{rgb}{0, 0, 11}
\newcommand{\emphic}[2]{%
   \textcolor{blue25}{%
      \textbf{\emph{#1}}}%
   \index{#2}}

\newcommand{\emphi}[1]{\emphic{#1}{#1}}

\newcommand{\etal}{\textit{et~al.}\xspace}

\providecommand{\Matousek}{Matou{\v s}ek\xspace}
\providecommand{\FejesToth}{Fejes-T\'oth\xspace}

\newcommand{\eps}{{\varepsilon}}%

\newcommand{\Arr}{{\cal A}}%

\def\C{{\cal C}}%
\def\P{{\cal P}}%
\def\Lift{{\bf \uparrow}}%
\def\crX{{\bf cr}}%
\def\level{{\rm level}}%

\newcommand{\si}[1]{#1}%

\def\marrow{\marginpar[\hfill$\longrightarrow$]{$\longleftarrow$}}
\newcommand{\micha}[1]{\textsc{(Micha says: } %
   \marrow\textsf{\textcolor{blue}{#1})}}
\newcommand{\michaH}[1]{\micha{#1}}

\def\haim#1{\textsc{(Haim says: }\marrow\textsf{#1})}
\newcommand{\sariel}[1]{{\textsc{(Sariel says: }\marrow%
      \textcolor{red}{\textsf{#1}})}}%

\remove{%
   \def\micha#1{}%
   \def\sariel#1{}%
   \def\haim#1{}%
}

\newcommand{\HLinkShort}[2]{\hyperref[#2]{#1\ref*{#2}}}
\newcommand{\HLink}[2]{\hyperref[#2]{#1~\ref*{#2}}}
\newcommand{\HLinkPage}[2]{\hyperref[#2]{#1~\ref*{#2}%
      $_\text{p\pageref{#2}}$}}
\newcommand{\HLinkPageOnly}[1]{\hyperref[#1]{Page~\refpage*{#1}%
      $_\text{p\pageref{#1}}$}}

\newcommand{\HLinkSuffix}[3]{\hyperref[#2]{#1\ref*{#2}{#3}}}
\newcommand{\HLinkPageSuffix}[3]{\hyperref[#2]{#1\ref*{#2}%
      #3$_\text{p\pageref{#2}}$}}

\newcommand{\seclab}[1]{\label{sec:#1}}
\newcommand{\secref}[1]{\HLink{Section}{sec:#1}}

\providecommand{\eqlab}[1]{}%
\renewcommand{\eqlab}[1]{\label{equation:#1}}
\newcommand{\Eqref}[1]{\HLinkSuffix{Eq.~(}{equation:#1}{)}}

\newcommand{\thmlab}[1]{{\label{theo:#1}}}
\newcommand{\thmref}[1]{\HLink{Theorem}{theo:#1}}

\newcommand{\invlab}[1]{\label{item:#1}}

\newcommand{\invref}[1]{\HLinkSuffix{(I}{item:#1}{)}}

\newcommand{\figlab}[1]{\label{fig:#1}}
\newcommand{\figref}[1]{\HLink{Figure}{fig:#1}}

\newcommand{\corlab}[1]{\label{cor:#1}}
\newcommand{\corref}[1]{\HLink{Corollary}{cor:#1}}%

\renewcommand{\Re}{\mathbb{R}}%
\newcommand{\reals}{\Re}%

\newcommand{\brc}[1]{\left\{ {#1} \right\}}
\newcommand{\ceil}[1]{\left\lceil {#1} \right\rceil}
\newcommand{\pth}[2][\!]{\mleft({#2}\mright)}%
\newcommand{\cardin}[1]{\left| {#1} \right|}%

\newcommand{\bd}{\partial}%

\newcommand{\CHX}[1]{\mathrm{CH}\pth{#1}}%
\def\D{{\cal D}}%
\def\T{{\cal T}}%
\def\U{{\cal U}}%

\newcommand{\lemlab}[1]{\label{xlemma:#1}}
\newcommand{\lemref}[1]{\HLink{Lemma}{xlemma:#1}}%

\newcommand{\R}{\mathcal{R}}%
\renewcommand{\th}{th\xspace}

\definecolor{sarielcolor}{rgb}{0.4, 0, 0.1}

\newcommand{\IncludeGraphics}[2][]{%
   \typeout{}%
   \typeout{Graphics: #2}%
   \typeout{\ includegraphics[#1]{#2}}%
   \includegraphics[#1]{#2}
   \typeout{}%
}

\begin{document}

\title{Approximating the $k$-Level in Three-Dimensional Plane
   Arrangements%
   \thanks{%
      A preliminary version of this paper appeared in {\it
         \si{Proc.~27th Annu. ACM-SIAM Sympos. Discrete Algs.}}
      (SODA), 2016, 1193--1212 \cite{hks-akltd-16}.
      Work by Sariel Har-Peled was partially supported by NSF AF
      awards CCF-1421231 and CCF-1217462.  %
      Work by Haim Kaplan was partially supported by grant 1161/2011
      from the German-Israeli Science Foundation, by grant 822/10 from
      the Israel Science Foundation, and by the Israeli Centers for
      Research Excellence (I-CORE) program (center no.~4/11).
      Work by Micha Sharir has been supported by Grant 2012/229 from
      the U.S.-Israel Binational Science Foundation, by Grant 892/13
      from the Israel Science Foundation, by the Israeli Centers for
      Research Excellence (I-CORE) program (center no.~4/11), and by
      the Hermann Minkowski--MINERVA Center for Geometry at Tel Aviv
      University.  }%
}

\author{%
   Sariel Har-Peled\thanks{%
      Department of Computer Science, University of Illinois, 201
      N.~Goodwin Avenue, Urbana, IL, 61801, USA. E-mail: {\tt
         sariel@illinois.edu}; {\tt \si{url}: http://sarielhp.org/} }%
   \and%
   Haim Kaplan%
   \thanks{%
      School of Computer Science, Tel Aviv University, Tel Aviv 69978,
      Israel. E-mail: {\tt haimk@post.tau.ac.il }}%
   \and%
   Micha Sharir%
   \thanks{%
      School of Computer Science, Tel Aviv University, Tel Aviv 69978,
      Israel.  E-mail: {\tt michas@post.tau.ac.il }}%
}


\date{\today}

\maketitle

\begin{abstract}
    Let $H$ be a set of $n$ non-vertical planes in three dimensions,
    and let $r<n$ be a parameter.  We give a simple alternative proof
    of the existence of an $O(1/r)$-cutting of the first $n/r$ levels
    of $\Arr(H)$, which consists of $O(r)$ semi-unbounded vertical
    triangular prisms.  The same construction yields an approximation
    of the $(n/r)$-level by a terrain consisting of $O(r/\eps^3)$
    triangular faces, which lies entirely between the levels $n/r$ and
    $(1+\eps)n/r$.  The proof does not use sampling, and exploits
    techniques based on planar separators and various structural
    properties of levels in three-dimensional arrangements and of
    planar maps. The proof is constructive, and leads to a simple
    randomized algorithm, that computes the approximating terrain in
    $O(n + r \eps^{-6} \log^3 r)$ expected time.  An application of
    this technique allows us to mimic and extend \Matousek's construction of
    cuttings in the plane~\cite{m-cen-90}, to obtain a similar
    construction of a ``layered'' $(1/r)$-cutting of the entire
    arrangement $\Arr(H)$, of optimal size $O(r^3)$.  Another
    application is a simplified optimal approximate range counting
    algorithm in three dimensions, competing with that of Afshani and
    Chan~\cite{ac-arcd-09}.
\end{abstract}


\section{Introduction}

\paragraph{A tribute to Jirka \Matousek.}

We were very fortunate to have Jirka as a friend and
colleague.  He has entered our community in the late 1980's, and has
been a giant lighthouse ever since, showing us the way into new
discoveries, solving mysteries for us, and just providing us with new
tools, ideas, and techniques, that have made our work much more
interesting and productive. He has been everywhere, making seminal
contributions to so many topics in computational and discrete geometry
(and to other fields too).  We have been avid readers of his many
books, most notably {\it Lectures on Discrete Geometry}, and have been
admiring his clear yet precise style of exposition and presentation.
We have also learned to appreciate his personality, his dry but
touching sense of humor, his love for nature, his infinite devotion to
science on one hand, and to his family and friends on the other
hand. His departure has been painful to us, and we will miss him
badly. We thank you, Jirka, for all the gifts you gave us, and may
your soul be blessed.

This paper is about a topic that Jirka has worked on, rather
extensively, during the early 1990s, concerning \emph{cuttings} and
related techniques for decompositions of arrangements or of point
sets, and their applications to range searching and other algorithmic
and combinatorial problems in geometry. In particular, in 1992 he has
written a seminal paper on ``Reporting points in
halfspaces''~\cite{m-rph-92}, where he introduced and analyzed
\emph{shallow cuttings}, a technique that had many applications during
the following decades.

In a later paper, following his earlier work~\cite{m-cen-90} (probably
his first entry into computational geometry), Jirka \cite{m-ocfcp-98}
presented a construction of $(1/r)$-cuttings, for a set of lines in the
plane, with $\leq 8r^2+6r+4$ cells. This construction uses, as a basic
building block, a strikingly simple procedure for approximating a level
in a line arrangement: Since a specific level is an $x$-monotone polygonal
chain, one can pick every $q$-th vertex, for $q \approx n/r$, and connect
these vertices consecutively to form an approximate level, which is at
crossing distance at most $q/2$ from the original level. As is well known,
this construction is asymptotically optimal for any arrangement of lines
in general position. This elegant level approximation algorithm, in two
dimensions, raises the natural question of whether one can approximate
a level in three dimensions for a given set of planes, by an $xy$-monotone
polyhedral terrain constructed directly, in an analogous manner, from the
original level.

This paper provides an affirmative answer to this question, thereby pushing Jirka's
work further, for the special case of three-dimensional arrangements of
planes. It refines the shallow cuttings technique of~\cite{m-rph-92},
and applies it to obtain cleaner and more efficient solutions for several
related problems. Our new scheme for approximating a level by a terrain,
while significantly more involved than Jirka's two-dimensional construction,
still echoes and generalizes his basic idea of ``shortcutting'' the
original level by a coarser triangular mesh (instead of a simplified
polygonal chain) spanned by selected vertices of the level.

\paragraph{Cuttings.}
Let $H$ be a set of $n$ (non-vertical) hyperplanes in $\Re^d$, and let
$r<n$ be a parameter. A \emph{$(1/r)$-cutting} of the arrangement
$\Arr(H)$ is a collection of pairwise openly disjoint simplices (or
other regions of constant complexity) such that the closure of their union covers $\Re^d$, and
each simplex is crossed (meets in its interior) at most $n/r$
hyperplanes of $H$.

Cuttings have proved to be a powerful tool for a
variety of problems in discrete and computational geometry, because
they provide an effective divide-and-conquer mechanism for tackling
such problems; see Agarwal \cite{a-gpia-91i} for an early
survey. Applications include a variety of range searching
techniques~\cite{ae-rsir-99}, partition trees \cite{m-ept-92},
incidence problems involving points and lines, curves, and
surfaces~\cite{cegsw-ccbac-90}, and many more.

The first (albeit suboptimal) construction of cuttings is due to
Clarkson \cite{c-narsc-87}. This concept was formalized later on by
Chazelle and Friedman~\cite{cf-dvrsi-90}, who gave a sampling-based
construction of optimal-size cuttings (see below). An optimal
deterministic construction algorithm was provided by Chazelle
\cite{c-chdc-93}.  \Matousek \cite{m-ocfcp-98} studied the number of
cells in a $(1/r)$-cutting in the plane (see also \cite{h-cctp-00}).
See Agarwal and Erickson~\cite{ae-rsir-99} and Chazelle~\cite{Chaz04}
for comprehensive reviews of this topic.

To be effective, it is imperative that the number of simplices in the
cutting be asymptotically as small as possible. Chazelle and
Friedman~\cite{cf-dvrsi-90} were the first to show the existence of a
$(1/r)$-cutting of the entire arrangement of $n$ hyperplanes in
$\Re^d$, consisting of $O(r^d)$ simplices, which is asymptotically the
best possible bound. (We note in passing that cuttings of optimal size
are not known for arrangements of (say, constant-degree algebraic)
surfaces in $\Re^d$, except for $d=2$, where the known bound,
$O(r^2)$, is tight, and for $d=3,4$, where nearly tight bounds,
nearly cubic and quartic in $r$, respectively, are
known~\cite{cegs-sessr-91, k-atubv-04, ks-csc-05}.)

For additional works related to cuttings and their applications, see
\cite{ac-ohrrt-09,act-dreod-14,ak-odsct-14,a-pal1e-90,a-pal2a-90,
   a-idapa-91,
   aacs-lalsp-98,m-ept-92,m-rsehc-93,ct-oda23-15,ak-odsct-14,h-cctp-00,r-orrrs-99}.

\paragraph{Shallow cuttings.}
The \emph{level} of a point $p$ in the arrangement $\Arr(H)$ of $H$ is
the number of hyperplanes lying vertically below it (that is, in the
$(-x_d)$-direction).  For a given parameter $0\le k\le n-1$,
the \emph{$k$-level}, denoted as $L_k$, is the
closure of all the points that lie on some hyperplane of $H$ and are
at level exactly $k$, and the $(\le k)$-level, denoted as $L_{\le k}$,
is the union of all the $j$-levels, for $j=0,\ldots,k$.  A collection
of pairwise openly disjoint simplices such that the closure of their union covers $L_{\leq k}$,
and such that each simplex is crossed at most $n/r$ hyperplanes of $H$, is
called a \emph{$k$-shallow $(1/r)$-cutting}.  Naturally, the
parameters $k$ and $r$ can vary independently, but the interesting
case, which is the one that often arises in many applications, is the
case where $k=\Theta(n/r)$.  In fact, shallow cuttings for any value
of $k$ can be reduced to this case---see Chan and
Tsakilidis~\cite[Section 5]{ct-oda23-15}.

In his seminal paper on reporting points in
halfspaces~\cite{m-rph-92}, \Matousek has proved the existence of
small-size shallow cuttings in arrangements of hyperplanes in any
dimension, showing that the bound on the size of the cutting can be
significantly improved for shallow cuttings. Specifically, he has
shown the existence of a $k$-shallow $(1/r)$-cutting, for $n$
hyperplanes in $\Re^d$, whose size is
$O\pth{q^{\lceil d/2\rceil}r^{\lfloor d/2\rfloor}}$, where
$q=k(r/n)+1$. For the interesting special case where $k=\Theta(n/r)$,
we have $q=O(1)$ and the size of the cutting is
$O\left(r^{\lfloor d/2\rfloor}\right)$, a significant improvement over
the general bound $O(r^d)$. (For example, in three dimensions, we get
$O(r)$ simplices, instead of $O(r^3)$ simplices for the whole
arrangement.)  This has lead to improved solutions of many range
searching and related problems.

In his paper, \Matousek presented a deterministic algorithm that can
construct such a shallow cutting in polynomial time; the running time
improves to $O(n \log r)$ but only when $r$ is small, i.e.,
$r < n^\delta$ for a sufficiently small constant $\delta$ (that depends
on the dimension $d$).  Later,
Ramos \cite{r-orrrs-99} presented a (rather complicated) randomized
algorithm for $d = 2,3$, that constructs a hierarchy of shallow
cuttings for a geometric sequence of $O(\log n)$ values of $r$,
where for each $r$ the corresponding cutting is a $(1/r)$-cutting
of the first $\Theta(n/r)$ levels of $\Arr(H)$. Ramos's algorithm runs
in $O(n\log n)$ total expected time.  Recently, Chan and Tsakalidis
\cite{ct-oda23-15} provided a deterministic $O(n \log r)$-time
algorithm for computing an $O(n/r)$-shallow $(1/r)$-cutting. Their
algorithm can also construct a hierarchy of shallow cuttings for a
geometric sequence of $O(\log n)$ values of $r$, as above, in $O(n\log n)$
deterministic time.  Interestingly, they use \Matousek's theorem on the
existence of an $O(n/r)$-shallow $(1/r)$-cutting of size $O(r)$
in the analysis of their algorithm.


Each simplex $\Delta$ in the cutting has a {\em conflict list}
associated with it, which is the set of hyperplanes intersecting
$\Delta$. The algorithms mentioned above for computing cuttings also
compute the conflict lists associated with the simplices of the
cutting. Alternatively, given the cutting, one can produce the conflict
lists in $O(n\log r)$ time using a result of Chan \cite{c-rshrr-00},
as we outline in \secref{sec:efficient}.

\Matousek's proof of the existence of small-size shallow cuttings, as
well as subsequent studies of this technique, are fairly
complicated. They rely on random sampling, combined with a clever
variant of the so-called exponential decay lemma of~\cite{cf-dvrsi-90},
and with several additional (and rather intricate) techniques.

\paragraph{Approximating a level.} %
An early study of \Matousek \cite{m-cen-90} gives a construction of a
$(1/r)$-cutting of small (optimal) size in arrangements of lines in
the plane.  The construction chooses a sequence of $r$ levels, $n/r$
apart from one another, and approximates each of them by a coarser
polygonal line, by choosing every $n/(2r)$-\th vertex of the level, and
by connecting them by an $x$-monotone polygonal path.  Each
approximate level does not deviate much from its original level, so
they remain disjoint from one another. Then, partitioning the region
between every pair of consecutive approximate levels into vertical
trapezoids produces a total of $O(r^2)$ such trapezoids, each crossed
by at most $O(n/r)$ lines.

It is thus natural to ask whether one can approximate, in a similar
fashion, a $k$-level of an arrangement of planes in 3-space. This is
significantly more challenging, as the $k$-level is now a polyhedral
terrain, and while it is reasonably easy to find a good (suitably small)
set of vertices that ``represent'' this level (in an appropriate sense,
detailed below), it is less clear how to triangulate them effectively
to form an $xy$-monotone terrain, such that (i) none of its triangles is
crossed by too many planes of $H$, and (ii) it remains close to the original
level.  To be more precise, given $k$ and $\eps>0$, we want to find a
polyhedral terrain with a small number of faces, which lies entirely
between the levels $k$ and $(1+\eps)k$ of $\Arr(H)$. A simple tweaking of
\Matousek's technique produces such an approximation in the planar
case, but it is considerably more involved to do it in 3-space.

Algorithms for terrain approximation, such as in \cite{ad-eats-97}, do
not apply in this case, as they have a quadratic blowup in the output
size, compared to the optimal approximation. Also, they are not geared
at all to handle our measure of approximation (in terms of lying close
to a specified level, in the sense that no point on the approximation
is separated by too many planes from the level).

Such an approximation to the $k$-level, whose size is optimal up to
polylogarithmic factors, can be obtained by using a
\emph{relative-approximation} sample of the planes, and by extracting
the appropriate level in the sample \cite{hs-rag-11}.  A more natural
approach, of using the triangular faces of an optimal-size shallow
cutting to form an approximate $k$-level, seems to fail in this case,
as the shallow cutting is in general just a collection of simplices,
stacked on top of one another, with no clearly defined $xy$-monotonicity.
Such a monotonicity is obtained in
Chan~\cite{c-ldlpv-05}, by replacing a standard shallow cutting by a suitable
upper convex hull of its simplices. 
However, the resulting cuttings do not lead to a sharp approximation of the
level, of the sort we seek.

In short, a simple,
effective, and optimal technique for approximating a level in three
dimensions (let alone in higher dimensions) does not follow easily
from existing techniques.

An additional advantage of such an approximation is that it
immediately yields a simply-shaped shallow cutting of the first $k$
levels of $\Arr(H)$, by replacing each triangle $\Delta$ of the
approximate level by the vertical semi-unbounded triangular prism
$\Delta^*$ having $\Delta$ as its top face, and consisting of all
points that lie on or vertically below $\Delta$.  Such a cutting (by
prisms) has already been constructed by Chan~\cite{c-ldlpv-05}, but it
does not yield (that is, come from) a $(1+\eps)$-approximation to the
level. Such a shallow cutting, by vertical semi-unbounded triangular
prisms, was a central tool in Chan's algorithm for dynamic convex
hulls in three dimensions~\cite{c-dds3c-10}.

Thus, resolving the question of approximating the $k$-level by an
$xy$-monotone terrain of small, optimal size is not a mere technical
issue, but rather a tool that will shed more light on the geometry of
arrangements of planes in three dimensions, and that has applications
to a variety of problems. For example, it yields an efficient
algorithm for approximating the level of a point in an arrangement of
planes in $\Re^3$, which is the dual version of approximate halfspace
range counting---see \secref{approx:r:count} for details. (Afshani and
Chan \cite{ac-arcd-09} present a similar approach to approximating the
level which is slightly more involved, as they do not have the desired
terrain property.)

\subsection{Our results}

In this paper we give an alternative, simpler and constructive proof
of the existence of optimal-size shallow cuttings in a
three-dimensional plane arrangement, by vertical semi-unbounded
triangular prisms.  With a bit more care, the construction yields an
optimal-size approximate level, as discussed above. Specifically,
given $r$ and $\eps$, one can approximate the $(n/r)$-level in an
arrangement of $n$ non-vertical planes in $\reals^3$, by a polyhedral
terrain of complexity $O(r/\eps^3)$, that lies entirely between the
levels $n/r$ and $(1+\eps)n/r$. The same construction works for any
values of the level $k$ and the parameter $r \leq n/k$, with a
somewhat more involved bound on the complexity of the approximation.

The construction does not use sampling, nor does it use the
exponential decay lemma of \cite{cf-dvrsi-90,m-rph-92}.  It is based
on the planar separator theorem of Lipton and
Tarjan~\cite{lt-stpg-79}, or, more precisely, on recent
separator-based decomposition techniques of planar maps, as in Klein
\etal~\cite{kms-srsdp-13} (see also Frederickson \cite{f-faspp-87}),
and on several insights into the structure and properties of levels in
three dimensions and of planar maps, which we believe to be of
independent interest.

As what we believe to be an interesting application of our technique,
we extend \Matousek's construction~\cite{m-cen-90} of cuttings in
planar arrangements to three dimensions. That is, we construct a
``layered'' $(1/r)$-cutting of the entire arrangement $\Arr(H)$ of a
set $H$ of $n$ non-vertical planes in $\Re^3$, of optimal size $O(r^3)$, by
approximating each level in a suitable sequence of levels, and then by
triangulating each layer between consecutive levels in the
sequence. The analysis becomes considerably more involved in three
dimensions, and requires several known but interesting and fairly
advanced properties of plane arrangements.

Another application of our technique is to approximate range counting.
Specifically, we show how to preprocess a set $H$ of $n$ non-vertical
planes in $\reals^3$, and a prescribed error parameter $\eps>0$, in
near-linear time (in $n$), into a data structure of size
$O(n/\eps^{8/3})$, so that, given a query
point $q\in\Re^3$, we can compute the number of planes of $H$ lying
below $q$, up to a factor $1\pm\eps$, in $O(\log(n/(\eps k)))$
expected time.  As noted, this competes with Afshani and Chan's
technique \cite{ac-arcd-09}. The general approach is similar in both
solutions, but our solution is somewhat simpler, due to the
availability of approximating terrains, and the dependence on $\eps$
in our solution is explicit and reasonable (this dependence is not
given explicitly in~\cite{ac-arcd-09}).

The thrust of this paper is thus to show, via alternative, simpler, and
more geometric methods, the existence of cuttings and approximate
levels of optimal size. The proofs are constructive, but naive
implementations thereof would be rather inefficient.  Nevertheless,
using standard random sampling techniques, we can obtain simple
randomized algorithms that perform (suitable variants of) these
constructions efficiently.  Specifically, they run in near-linear
expected time (which becomes linear when $r$ is not too
large).

\paragraph{Sketch of our technique.}
The $k$-level in a plane arrangement in three dimensions is an
$xy$-monotone polyhedral terrain. After triangulating each of its faces,
its $xy$-projection forms a (straight-edge) triangulated biconnected
planar map. Since the average complexity of the first $k$ levels is
$O(nk^2)$ (see, e.g., \cite{cs-arscg-89}), we may assume, by moving from
a specified level to a nearby one, that the complexity of our level is $O(nk)$.
The decomposition techniques of planar graphs mentioned above (as in
\cite{kms-srsdp-13}) allow us to partition the level into $O(n/k)$
clusters, where each cluster has $O(k^2)$ vertices and  $O(k)$
boundary vertices (vertices that also belong to other clusters). In the terminology of
\cite{kms-srsdp-13}, this is a \emph{$k^2$-division} of the graph.
Each such cluster, projected to the $xy$-plane, is a polygon with
$O(k)$ boundary edges (and with $O(k^2)$ interior edges).
We show that, replacing each such projected polygon by
its convex hull results in a collection of $O(n/k)$ convex
\emph{pseudo-disks}, namely, each hull is (trivially) simply connected,
and the boundaries of any pair of hulls intersect at most twice.
Moreover, the decomposition has the property that, for
each triangle $\Delta$ that is fully contained in such a
pseudo-disk, lifting its vertices back to the $k$-level yields a
triple of points that span a triangle $\Delta'$ with a small number of
planes crossing it, so it lies close to the $k$-level.

An old result of Bambah and Rogers~\cite{br-cpcs-52}, proving a
statement due to L. \FejesToth, and reviewed in \cite[Lemma 3.9]{pa-cg-95}
(and also briefly below), shows that a union of $m$ convex pseudo-disks that
covers the plane induces a triangulation of the plane by $O(m)$
triangles, such that each triangle is fully contained inside one of
the pseudo-disks.  (As a matter of fact, it shows that each
pseudo-disk can be
shrunk into convex polygon so that these polygons are
 pairwise openly disjoint, with the same union, and the total number of
edges of the polygons is at most $6m$; the desired triangulation is
obtained by simply triangulating, arbitrarily, each of these polygons.)  Lifting
(the vertices of) this triangulation to the $k$-level, with a
corresponding lifting of its triangular faces, results in the desired
terrain approximating the level.  A significant technical
contribution of this paper is to provide an alternative proof of this
result. The original proof in \cite{br-cpcs-52} appears to be fairly
involved, although its presentation in \cite{pa-cg-95} is
simplified. Still, it does not seem to lead to a sufficiently
efficient construction. Our proof in contrast does lead to such a
construction, as described in \secref{triangulation}.

A shallow cutting of the first $k$ levels is obtained by simply
replacing each triangle $\Delta$ in the approximate level by the
semi-unbounded vertical prism of points lying below $\Delta$.

\paragraph{Confined triangulations.}
The idea of decomposing the union of objects (pseudo-disks here) into
pairwise openly disjoint simply-shaped fragments, each fully
contained in some original object, is
implicit in algorithms for efficiently computing the union of objects;
see the work of Ezra \etal \cite{egs-suicu-04}, which was in turn
inspired by Mulmuley's work on hidden surface removal
\cite{m-eahsr-94}.  Mustafa \etal \cite{mrr-sahsg-14} use a more
elaborate version of such a decomposition, for situations where the
objects are weighted. While these decompositions are useful for a
variety of applications, they still suffer from the problem that the
complexity of a single region in the decomposition might be
arbitrarily large. In contrast, the triangulation scheme that we use
(following \cite{br-cpcs-52}) is simpler, optimal, and independent of
the complexity of the relevant pseudo-disks.  We are pleased that this
nice property of convex pseudo-disks is (effectively) applicable to the
problems studied here, and expect it to have many additional potential
applications.

In particular, we extend our analysis, and
show that such a decomposition exists for arbitrary
convex shapes, with the number of pieces being proportional to the union
complexity, and with each region being a triangle or a cap (i.e., the
intersection of an input shape with a halfplane). This provides a
representation of ``most'' of the union by triangles, where the more
complicated caps are only used to fill in the ``fringe'' of the union
(and are absent when the union covers the entire plane, as in
\cite{br-cpcs-52}). We believe that this triangulation could be useful
in practice, in situations where, given a query point $q$, one wants to
decide whether $q$ is inside the union, and if so, provide a witness
shape that contains $q$. For this, we simply locate the triangle
in our triangulation that contains $q$, from which the desired
witness shape is immediately available. This is significant in
situations where deciding whether a point belongs to an input shape
is considerably more expensive than deciding whether it lies inside
a triangle.

\paragraph{Paper organization.}
We start by presenting the construction of the confined triangulation in
\secref{triangulation}. We then describe the construction of approximate levels,
and the construction of shallow cuttings that it leads to, in \secref{approx_level}.
We then present applications of our results in \secref{applications}.
Specifically, in \secref{layered:cutting} we show how to build a layered cutting of the
whole arrangement, and in \secref{approx:r:count} we show how to answer
approximate range counting queries for halfspaces.


\section{Triangulating the union of convex shapes}%
\seclab{triangulation}%

In this section we show that, given a finite collection of $m$ convex
pseudo-disks covering the plane, one can construct a triangulation of
the plane, consisting of $O(m)$ triangles, such that each triangle is
contained in a single original pseudo-disk---see
\thmref{triangulate:union} below for details. Our result can be
extended to situations where the union of the pseudo-disks is not the
entire plane; see below.  This claim is a key ingredient in our
construction of approximate $k$-levels, detailed in
\secref{approx_level}, but it is not new, as it is an immediate
consequence of an old result of Bambah and Rogers~\cite{br-cpcs-52}
(proving a statement by L. \FejesToth), whose proof is sketched below.

\smallskip%

\newcommand{\KK}{\mathcal{K}}%

\parpic[r]{\IncludeGraphics{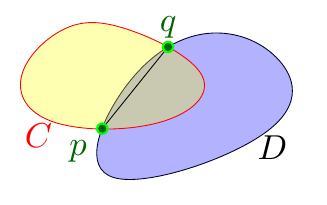}}%

\noindent%
\textbf{Bambah and Rogers' proof.}  For the sake of completeness, we briefly
sketch the proof of Bambah and Rogers (as presented in Pach and
Agarwal~\cite[Lemma 3.9]{pa-cg-95}). Let $\KK$ be a collection of $m$
polygonal convex pseudo-disks in the plane, and assume, for simplicity, that
their union is a triangle $T$ (extending this simplest scenario to the
more general case is straightforward). We may also assume that no pseudo-disk
of $\KK$ is contained in the union of the other regions of $\KK$, as one can
simply throw away any such redundant pseudo-disk. Finally, since the construction
will create regions with overlapping boundaries, we use the more general
definition of pseudo-disks, requiring, for each pair $C,D\in\KK$, that
$C\setminus D$ and $D\setminus C$ are both connected.

\parpic[r]{\IncludeGraphics[page=2]{figs/pd}}%

Let $C$ and $D$ be two pseudo-disks of $\KK$, such that the
common intersection $\mathrm{int}(C) \cap \mathrm{int}( D )$ of their interiors is nonempty
and minimal in terms of containment (that is, it does not contain any
other such intersection). Let $p$ and $q$ be the two intersection points
of $\bd C$ and $\bd D$ (assume for simplicity that $\bd C$ and $\bd D$
do not overlap, making $p$ and $q$ well defined). Cut $C$ and $D$ along
the segment $pq$, and let $C' \subseteq C$ and $D' \subseteq D$ be the
two resulting pieces whose union is $C \cup D$.
Let $\KK' = \pth{\KK \setminus \brc{ C, D}} \cup \brc{C',D'}$.
The claim is that $\KK'$ is a collection of $m$ pseudo-disks covering $T$.

\parpic[r]{\IncludeGraphics[page=3]{figs/pd}}%

Indeed, consider a pseudo-disk $E \in \KK'$ other than $C'$, $D'$.
We need to show that $E \setminus C'$ and $C' \setminus E$ are both
connected, and similarly for $E$ and $D'$. If $E$ contains $p$
(resp., $q$), then it is easy to verify, by convexity, that $E$ and
$C'$ are pseudo-disks, and similarly for $E$ and $D'$. Assume then
$E$ does not contain $p$ or $q$, but still intersects the segment $pq$.
By assumption, $E \setminus (C \cup D)$ is not empty, so we may assume,
without loss of generality, that $E$ intersects the boundary of
$\bd C \setminus D$. But then $E \cap D \subseteq C \cap D$, as
otherwise $E$ would intersect the boundary of $C$ in four points,
which is impossible. This in turn contradicts the minimality of $C \cap D$.

We thus replace $\KK$ by $\KK'$, and repeat this process till all the
pseudo-disks in the resultimg collection are pairwise interior disjoint.
At this point, $\KK$ is a pairwise openly disjoint cover of the triangle
$T$, by $m$ convex polygons (each contained inside its original pseudo-disk).
By Euler's formula, these polygons can be triangulated into $O(m)$ triangles
with the desired property.

\medskip

This elegant proof is significantly simpler than what follows, but
it does \emph{not} seem to lead to an efficient algorithm for constructing
the desired triangulation in near-linear running time.
We present here a different alternative (efficiently) constructive proof,
which leads to an $O(m\log m)$-time
algorithm for constructing the triangulation for a set of $m$
pseudo-disks, in a suitable model of computation.  (As an aside, we
also think that such a nice property deserves more than one proof.)
We also establish an extension of this result to more general convex
shapes.

\subsection{Preliminaries}

\parpic[r]{\IncludeGraphics[scale=0.9]{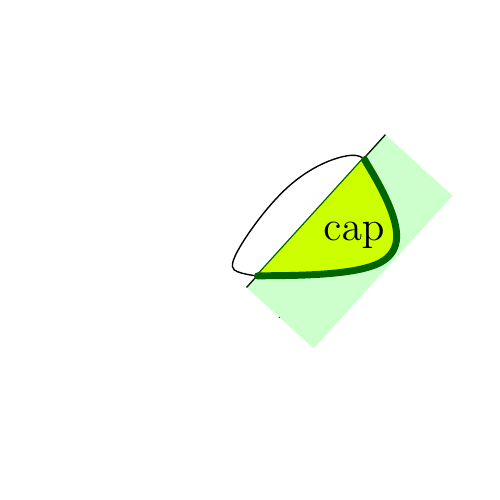}}%

\noindent%
The notion of a triangulation that we use here is slightly
non-standard, as it might be a triangulation of the entire plane, and
not just of the convex hull of some input set of points. As such, it
contains unbounded triangles, where the boundary of each such triangle
consists of one bounded segment and two unbounded rays (where the
segment might degenerate into a single point, in which case the
triangle becomes a wedge).

\parpic[r]{\IncludeGraphics[page=2]{figs/cap}}%

Given a convex shape $D$, a \emphi{cap} of $D$ is the region formed by
the intersection of $D$ with a halfplane. A \emphi{crescent} is a
portion of a cap obtained by removing from it a convex polygon that
has the base chord of the cap as an edge, but is otherwise contained in the interior of the cap.

\begin{defn}
    Given a collection $\D$ of convex shapes in the plane, a
    decomposition $\T$ of their union into pairwise openly disjoint regions is a
    \emphi{confined triangulation}, if
    \begin{inparaenum}[(i)]
        \item every region in $\T$ is either a triangle or a cap,
        and
        \item every such region is fully contained in one of the
        original input shapes.
    \end{inparaenum}
\end{defn}

See \figref{decomp:example} for an example of a confined
triangulation.

\begin{figure*}
    \begin{center}
        \begin{tabular}{ccc}
      \IncludeGraphics[page=1]{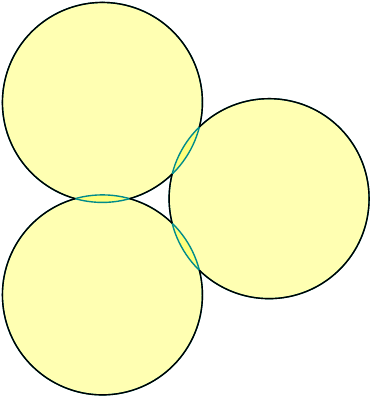}%
      &
        \qquad\qquad
        &
        \IncludeGraphics[page=2]{figs/triangulate_union}%
    \end{tabular}
    \end{center}
    \vspace{-0.7cm}%
    \caption{A union of three disks, and its decomposition into
       triangles and caps. Note that the decomposition computed by
       our algorithm is somewhat different for this case.}
    \figlab{decomp:example}%
\end{figure*}

\subsection{Construction}

We are given a collection $\D$ of $m$ convex pseudo-disks,
and our goal is to construct a confined triangulation for $\D$,
as described above, with $O(m)$ pieces. In what follows we consider both the case where
the union of $\D$ covers the plane, and the case where it does not.

\subsubsection{Painting the union from front to back.}
A basic property of a collection $\D$ of $m$ pseudo-disks is that the
combinatorial complexity of the boundary of the union
$\U:=\U(\D) = \bigcup_{C \in \D} C$ of $\D$ is at most $6m-12$, where
we ignore the complexity of individual members of $\D$, and just count
the number of intersection points of pairs of boundaries of members of
$\D$ that lie on $\bd\U$; see \cite{klps-ujrcf-86}. For convenience,
we also (i) include the leftmost and rightmost points of each $D\in\D$ in
the set of intersection points (if they lie on the union boundary),
thus increasing the complexity of the union by at most $2m$, and
(ii) assume general position of the pseudo-disks. In general, an
intersection point $v$ of a pair of boundaries is at \emph{depth $k$}
(of the arrangement $\Arr(\D)$ of $\D$) if it is contained in the
interiors of exactly $k$ members of $\D$. The boundary intersections
are thus at depth $0$, and a simple application of the Clarkson--Shor
technique~\cite{cs-arscg-89} implies that the number of boundary
intersection points that lie at depth $1$ is also $O(m)$. Hence there
exists at least one pseudo-disk $D\in\D$ that contains at most $c$
intersection points at depths $0$ or $1$ (including leftmost and
rightmost points of disks), for some suitable absolute constant
$c$. Clearly, these considerations also apply to any subset of $\D$.

\newcommand{\VZE}[3]{
   \begin{minipage}{0.225\linewidth}
       \smallskip%
       \centerline{%
          \IncludeGraphics[page=#1,%
          trim=0 #3 0 #2,clip,%
          width=1.04\linewidth]{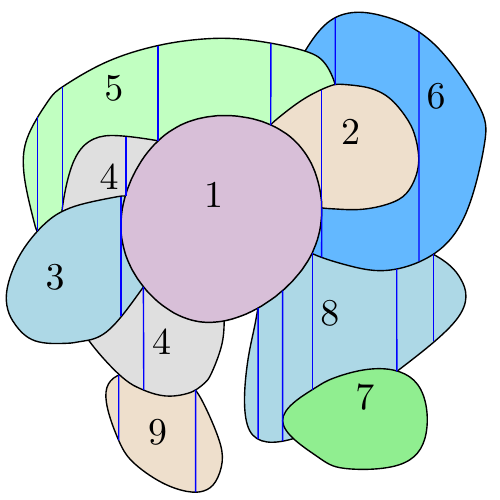}%
       }%
   \end{minipage}
} \newcommand{\VZ}[1]{
   \begin{minipage}{0.225\linewidth}
       \smallskip%
       {\IncludeGraphics[page=#1, width=1.00\linewidth]{figs/ptrp_2}}%
   \end{minipage}
}%
\begin{figure*}[p]
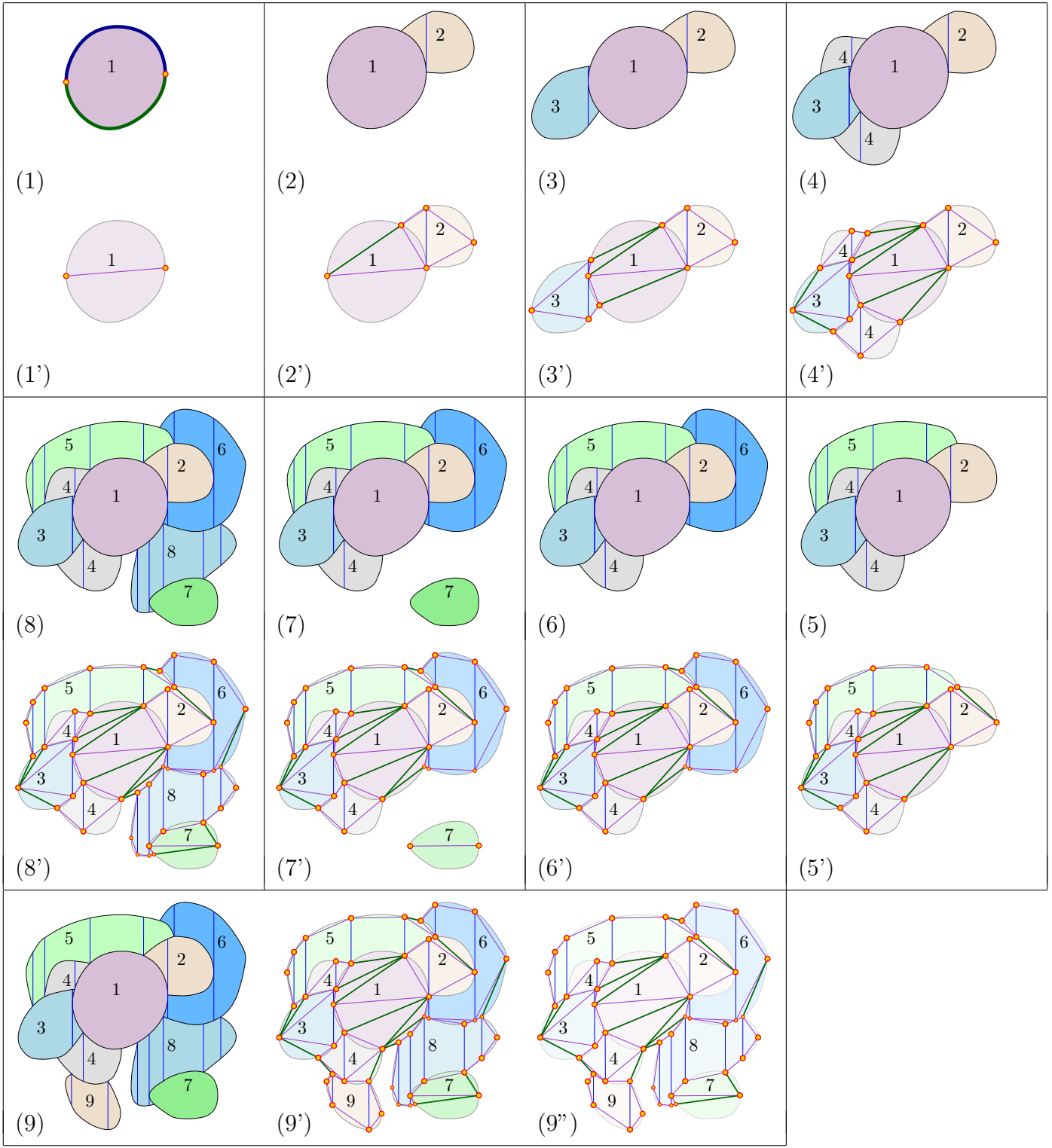

    \centerline{%
       \begin{tabular}{|l|l|l|l|}
         \hline
         \VZE{2}{23}{28}%
         &%
           \VZE{3}{23}{28}%
         &%
           \VZE{4}{23}{28}%
         &%
           \VZE{5}{23}{28}%
         \\
         (1) & (2) & (3) & (4) \\[0cm]%
         \VZE{11}{23}{28}%
         &%
           \VZE{12}{23}{28}%
         &
           \VZE{13}{23}{28}%
         &%
           \VZE{14}{23}{28}%
         \\
         (1') & (2') & (3') & (4') \\[0.1cm]%
      %
      %
         \hline
         \VZ{9}%
         &%
           \VZ{8}%
         &%
           \VZ{7}%
         &%
           \VZ{6}%
         \\[-0.5cm]%
         (8) & (7) & (6) & (5) \\%
         \VZ{18}%
         &
           \VZ{17}%
         &%
           \VZ{16}%
         &%
           \VZ{15}%
         \\[-0.5cm]%
         (8') & (7') & (6') & (5') \\[0.1cm]%
         \hline
         \multicolumn{1}{|c}{\VZ{10}}%
         &%
           \multicolumn{1}{c}{\VZ{19}}%
         &
           \multicolumn{1}{c|}{
           \VZ{20}}%
         \\[-0.4cm]%
         \multicolumn{1}{|l}{(9)}%
         &%
           \multicolumn{1}{l}{(9')}%
         &%
           \multicolumn{1}{l|}{(9'')}%
         \\[0.1cm]%
         \cline{1-3}%
       \end{tabular}%
    }
    \caption{A step-by-step illustration of the decomposition $\T$
       into pseudo-trapezoids and of the polygonalization of the
       union. See \secref{bridges}.
       \textbf{An animation of this figure is available online at
          \url{http://sarielhp.org/blog/?p=8920};} see also %
       \figref{figanim}. }
    \figlab{ptrp:3}
\end{figure*}

This allows us to order the members of $\D$ as $D_1,\ldots,D_m$, so
that the following property holds. Set $\D_i := \{D_1,\ldots,D_i\}$,
for $i=1,\ldots,m$.  Then $D_i$ contains at most $c$ intersection
points at depths $0$ and $1$ of $\Arr(\D_i)$. Equivalently, for each
$i$, the boundary of $D_i^0:= D_i\setminus\U(\D_{i-1})$ contains at
most $c$ intersection points.

To prepare for the algorithmic implementation of the construction in
this proof, which will be presented later, we note that this ordering
is not easy to obtain efficiently in a deterministic
manner. Nevertheless, a random insertion order (almost) satisfies the
above property: As we will show, the expected sum of the complexities of the regions
$D_i^0$, for a random insertion order, is $O(m)$. See later for more
details.

We thus have $\U(\D_j) = \bigcup_{i\le j} D_i^0$ (as an openly
disjoint union), for each $j$; for the convenience of presentation
(and for the algorithm to follow), we interpret this ordering as an
incremental process, where the pseudo-disks of $\D$ are inserted, one
after the other, in the order $D_1,\ldots,D_m$, and we maintain the
partial unions $\U(\D_j)$, after each insertion, by the formula
$\U(\D_j) = \U(\D_{j-1})\cup D_j^0$.

\subsubsection{Decomposing the union into vertical tr\-\si{apezoids}.}

Since the boundary of $D_i^0 = D_i\setminus\U(\D_{i-1})$ contains at
most $c$ intersection points, we can decompose $D_i^0$ into $O(1)$
\emph{vertical pseudo-trapezoids}, using the standard vertical
decomposition technique; see, e.g.,~\cite{sa-dsstg-95}.  Let $\T_j$ be
the collection of pseudo-trapezoids in the decomposition of
$\U(\D_j)$, collected from the decompositions of the regions $D_i^0$,
for $i=1,\ldots,j$, and let $V_j$ be the set of vertices of these
pseudo-trapezoids, each of which is either an intersection point
(more precisely, a boundary intersection or an $x$-extreme point) of $\Arr(\D_j)$, or an
intersection between some $\bd \D_i$ and a vertical segment erected
from an intersection point of $\Arr(\D_j)$.

Each of the pseudo-trapezoids in $\T_j$ is bounded by (at most) two
vertical segments, a portion of the boundary of a single pseudo-disk
as its top edge, and a portion of the boundary of (another) single
pseudo-disk as its bottom edge; see \figref{ptrp:3}.  We have
$D_1^0=D_1$, which we regard as a single pseudo-trapezoid, in which
the vertical sides degenerate to the leftmost and rightmost points of
$\bd D_1$; see \figref{ptrp:3}(1).  Note that in the vertical
decomposition of $D_i^0$ we split it by vertical segments through the
intersection points on its boundary, but not through vertices of
$V_{i-1}$ on $\bd D_i^0$ which are not intersection points of
$\Arr(\D)$. (Informally, these vertices are ``internal'' to
$\U(\D_{i-1})$, and are not ``visible'' from the outside.) See, e.g.,
\figref{ptrp:3}(4).  The set $V_i$ is obtained by adding to $V_{i-1}$
the vertices of the pseudo-trapezoids in the decomposition of $D_i^0$.

If $D_i^0$ is bounded then each pseudo-trapezoid $\tau$ in its
decomposition has a top boundary and a bottom boundary, but one or
both of the vertical sides may be missing (see, e.g.,
\figref{ptrp:3}(1) for the single pseudo-trapezoid $D_1^0=D_1$ and
\figref{ptrp:3}(3) for the left pseudo-trapezoid of $3$). From the
point of view of $\tau$, each of the top and bottom boundaries of
$\tau$ may be either convex (if it is a subarc of $\bd D_i$ on
$\bd D_i^0$), or concave (if it is part of the boundary of some
previously inserted pseudo-disk); If $D_i^0$ is not bounded then some
of the vertical pseudo-trapezoids covering $D_i^0$ will also be
unbounded and missing some of their boundaries.  Note that $D_i^0$ is
not necessarily connected; in case it is not connected we separately
decompose each of its connected components into vertical
pseudo-trapezoids in the above manner; see \figref{ptrp:3}(4).

At the end of the incremental process, after inserting all the $m$
pseudo-disks in $\D$, the pseudo-trapezoids in $\T:=\T_m$ cover $\U(\D)$,
which may or may not be the entire plane, and they
are pairwise openly disjoint.  By construction, each
pseudo-trapezoid in $\T$ is contained in a single pseudo-disk of
$\D$. Moreover, since the complexity of each $D_i^0$ is $O(1)$, the
total number of pseudo-trapezoids in $\T$ is $O(m)$.  So $\T$
possesses some of the properties that we want, but it is not a
triangulation.



\parpic[r]{%
   \IncludeGraphics[scale=1]{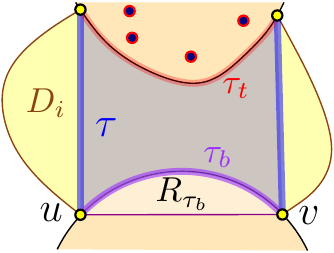}%
}%

\subsubsection{Polygonalizing the pseudo-trapezoids.} %

To get a triangulation, we associate a polygonal vertical
pseudo-trapezoid $\tau^*$ with each pseudo-trapezoid $\tau \in \T$.
We obtain $\tau^*$ from $\tau$ by replacing the bottom boundary
$\tau_b$ and the top boundary $\tau_t$ of $\tau$ by respective
polygonal chains $\tau_b^*$ and $\tau_t^*$, that are defined as
follows.\footnote{The term ``polygonal'' is somewhat misleading, as
   some of the boundaries of the pseudo-disks of $\D$ may also be
   polygonal.  To avoid confusion think of the boundaries of the
   pseudo-disks of $\D$ as smooth convex arcs (as drawn in the
   figures) even though they might be polygonal.} %
Let $D_i$ be the pseudo-disk during whose insertion $\tau$ was
created; in particular, $\tau\subseteq D_i^0$. Let $u$ and $v$ denote
the endpoints of $\tau_b$.  Consider the region $R_{\tau_b}$ between
$\tau_b$ and the straight segment $uv$; clearly, by the convexity of
$D_i$, $R_{\tau_b}$ is fully contained in $D_i$.  See figure on the
right.

\parpic[r]{\IncludeGraphics[page=2]{figs/shorten}}%
If $R_{\tau_b}$ contains no vertices of $V_i$, other than $u$ and $v$
(this will always be the case when $R_{\tau_b}\subseteq \tau$), we
replace $\tau_b$ by $\tau_b^*=uv$. Otherwise, we replace $\tau_b$ by
the chain $\tau_b^*$ of edges of the convex hull of
$V_i \cap R_{\tau_b}$, other than the edge $uv$.  We define $\tau_t^*$
analogously, and take $\tau^*$ to be the polygonal vertical
pseudo-trapezoid that has the same vertical edges as $\tau$, and its
top (resp., bottom) part is $\tau_t^*$ (resp., $\tau_b^*$). See figure
on the right.

\parpic[r]{\IncludeGraphics[page=3]{figs/shorten}}%

Note that, by construction, $\tau_b^*$ is a convex polygonal
chain. From the point of view of $\tau$, it is convex (resp., concave)
if and only if $\tau_b$ is convex (resp., concave).  (These statements
become somewhat redundant when $\tau_b^*$ is the straight segnment
$uv$.) An analogous property holds for $\tau_t^*$ and $\tau_t$.  We
denote the crescent-like region bounded by $\tau_b$ and $\tau_b^*$ by
$\overline{R}_{\tau_b}$; $\overline{R}_{\tau_t}$ is defined
analogously.  (Formally,
$\overline{R}_{\tau_b} = R_{\tau_b} \setminus CH(V_{i} \cap
R_{\tau_b})$ and
$\overline{R}_{\tau_t} = R_{\tau_t} \setminus CH(V_{i} \cap
R_{\tau_t})$.)  Let $\T_i^*$ be the set of polygonal vertical
pseudo-trapezoids associated in this manner with the pseudo-trapezoids
in $\T_i$.

\parpic[r]{\IncludeGraphics[page=1]{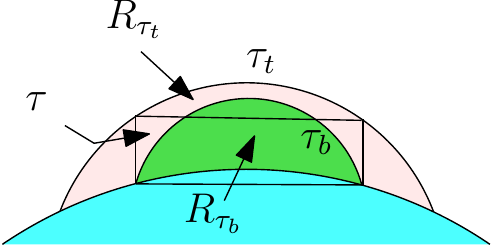}}%

Note that $R_{\tau_b}$ and $R_{\tau_t}$ need not be disjoint, as
illustrated in the figure on the right. Nevertheless, $\tau_b^*$ and $\tau_t^*$
cannot cross one another, as follows from Invariant \invref{i:2} that
we establish below (in \lemref{invariants}).  This implies that
$\tau^*$ is well defined. If $\tau_b^*$ and $\tau_t^*$ are not
disjoint then they may only be pinched together at common vertices, or
overlap in a single common connected portion (in the extreme case they
may be identical).

\parpic[r]{\IncludeGraphics[page=1]{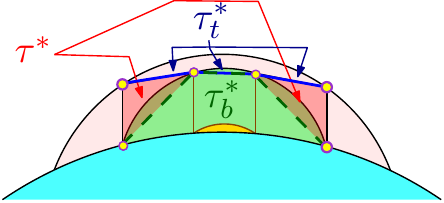}}%

This pinching or overlap, if it occurs, causes the
interior of $\tau^*$ to be disconnected (into at most two pieces,
as depicted in the figure to the right; it
may also be empty, as is the case for $D_1^0$, illustrated in
\figref{ptrp:3}(1)).

\subsubsection{Filling the cavities.}
\seclab{bridges}%
The insertion of $D_i$ may in general split some arcs of
$\bd\U(\D_{i-1})$ into subarcs, whose new endpoints are either points
of contact between $\bd D_i$ and $\bd\U(\D_{i-1})$, or endpoints of
vertical segments erected from other vertices of $D_i^0$.  This can be
seen all over \figref{ptrp:3}. For example, see the subdivision of the
top arc of $D_7$ caused by the insertion of $D_8$ in
\figref{ptrp:3}(8').  Some of these subarcs are boundaries of the new
pseudo-trapezoids of $D_i^0$ and thus do not belong to $\bd \U(\D_i)$,
and some remain subarcs of $\bd \U(\D_i)$. We refer to subarcs of the
former kind as \emph{hidden}, and to those of the latter kind as
\emph{exposed}. Note that, among the subarcs into which an arc of
$\bd\U(\D_{i-1})$ is split, only the leftmost and rightmost extreme subarcs can be exposed (this follows from the pseudo-disk
property of the objects of
$\D$).

We take each new exposed arc $\gamma$, with endpoints $u,v$, and apply
to it the same polygonalization that we applied above to $\tau_b$ and
$\tau_t$.  That is, we take the region $R_\gamma$ enclosed between
$\gamma$ and the segment $uv$, and define $\gamma^*$ to be either
$uv$, if $R_\gamma$ does not contain any vertex of $V_i$, or else the
boundary of $\CHX{R_\gamma\cap V_{i}}$, except for $uv$.  We note that
$\gamma^*$ is a convex polygonal chain that shares its endpoints with
$\gamma$, and denote the region enclosed between $\gamma$ and
$\gamma^*$ as $\overline{R}_{\gamma}$.

\parpic[r]{\IncludeGraphics[page=1]{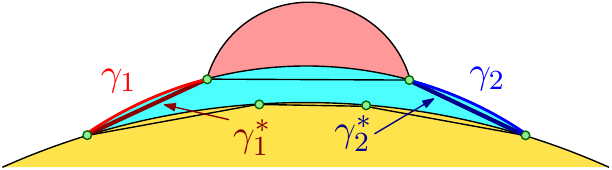}}%

Let $E_i$ denote the collection of all straight edges in the polygonal
boundaries of the pseudo-trapezoids in $\T_i^*$ and in the polygonal
chains $\gamma^*$ corresponding to new exposed subarcs $\gamma$ of
$\bd \U(\D_{j-1})$, $1\le j\le i$, which were created and
polygonalized when adding the corresponding pseudo-disk $D_j$.  See
figure on the right.

\subsubsection{Putting it all together}


\paragraph{When the pseudo-disks cover the plane.}
When the polygonalization process terminates, there are no more
regions $\overline{R}_\gamma$, for boundary arcs $\gamma$ of the union
(because there is no boundary), so we are left with a straight-edge
planar map $M$ with $E_m$ as its set of edges. (Invariant \invref{i:1}
in \lemref{invariants} below asserts that the edges in $E_m$ do not
cross each other.)  By Euler's formula, the complexity of $M$ is
$O(m)$.  We then triangulate each face of $M$, and, as the analysis in
the next subsection will show, obtain the desired triangulation.

\parpic[r]{\IncludeGraphics[page=3, scale=1.5]{figs/cap}}

\paragraph{The general case.} %
In general, the construction decomposes the union into (pairwise
openly disjoint) triangles and crescent regions.  To complete the
construction, we decompose each crescent region into triangles and
caps. A crescent region with $t\ge 2$ vertices on its concave boundary
can be decomposed into $t-2$ triangles and at most $t-1$ caps.  The
case $t=2$ is vacuous, as the crescent is then a cap, so assume that
$t\ge 3$.  To get such a decomposition, take an extreme edge of the
concave polygonal chain, and extend it till it intersects the convex
boundary of the crescent, at some point $w$, thereby chopping off a
cap from the crescent. We then create the triangles that $w$ spans
with all the concave edges that it sees, and then recurse on the
remaining crescent; see figure on the right. It is easily seen that
this results in $t-2$ triangles and at most $t-1$ caps, as
claimed. After this fix-up, we get a decomposition of the union into
triangles and caps. Here too, by Euler's formula, the complexity of $M$ is
$O(m)$.

\subsection{Analysis}

The correctness of the construction is established in the following lemma.

\begin{lemma} %
    \lemlab{invariants}%
    The pseudo-trapezoids in $\T_i^*$ and the edges of $E_i$ satisfy
    the following invariants: \smallskip%
    \begin{compactenum}[\rm \bf\quad({I}1)]
        \item \invlab{i:1} The segments in $E_i$ do not cross one
        another.

        \smallskip%
        \item \invlab{i:2} Each subarc $\gamma$ of $\bd \U(\D_{i})$
        with endpoints $u$ and $v$ has an associated convex polygonal
        arc $\gamma^* \subseteq E_i$ between $u$ and $v$. The chains
        $\gamma^*$ are pairwise openly disjoint, and their union forms
        the boundary of a polygonal region
        $\U_i^* \subseteq \U(\D_{i})$.

        \smallskip%
        \item \invlab{i:3} The pseudo-trapezoids in $\T_i^*$ are
        pairwise openly disjoint, and each of them is fully contained
        in some pseudo-disk of $\D_i$.

        \smallskip%
        \item \invlab{i:4}%
        $\U(\D_i) \setminus \underset{\tau^* \in \T_i^*}{\bigcup}
        \tau^*$
        consists of a collection of pairwise openly disjoint {\em holes}.
        Each hole is a region between two $x$-monotone
        convex chains or between two $x$-monotone concave chains, with
        common endpoints, where either both chains are polygonal, or
        one is polygonal and the other is a portion of the boundary of
        a single pseudo-disk that lies on $\bd\U(\D_i)$.
        (Each of the latter holes is a crescent-like region of the
        form $\overline{R}_{\tau_b}$, $\overline{R}_{\tau_t}$, for some
        trapezoid $\tau$, or $\overline{R}_{\gamma}$, for some exposed
        arc $\gamma$, as defined above.)
        The union of the holes of the latter kind (crescents) is $\U(\D_{i}) \setminus \U_i^*$.
        Each hole, of either kind, is fully contained in some
        pseudo-disk $D_j$, $j\le i$.
    \end{compactenum}
\end{lemma}

We refer to holes of the former (resp., latter) kind in \invref{i:4}
of the lemma as \emph{internal polygonal} holes (resp., \emph{external
half-polygonal} holes).

\begin{proof}
    We prove that these invariants hold by induction on $i$. The
    invariants clearly hold for $\T_1^*$ and $E_1$ after starting the process
    with $D_1^0=D_1$. Concretely, $\T_1^*$ consists of the single
    degenerate pseudo-trapezoid $uv$, where $u$ and $v$ are the
    leftmost and rightmost points of $\D_1$, respectively, and
    $E_1 = \{uv \}$. The (external half-polygonal) holes are the
    portions of $D_1$ lying above and below $uv$. It is obvious that
    \invref{i:1}--\invref{i:4} hold in this case.

    Suppose the invariants hold for $\T_{i-1}^*$ and $E_{i-1}$.  We
    first prove \invref{i:1} for $E_i$. By construction, the new edges
    in $E_i \setminus E_{i-1}$ form a collection of convex or concave
    polygonal chains, where each chain $\gamma^*$ starts and ends at
    vertices $u,v$ of either $\bd D_i^0$ or
    $\bd\U(\D_{i-1})$. Moreover, by construction, $u$ and $v$ are
    connected to one another by a single arc $\gamma$ of the
    respective boundary $\bd D_i^0$ or $\bd\U(\D_{i-1})$ ($\gamma$ is
    either an exposed or a hidden subarc of $\bd\U(\D_{i-1})$, or a
    subarc of $\bd D_i$ along $\bd D_i^0$), and the region
    $\overline{R}_\gamma$ between $\gamma$ and $\gamma^*$ does not
    contain any vertex of $V_i$ in its interior.

    Clearly, the edges in a single chain $\gamma^*$ do not cross one
    another.  Suppose to the contrary that an edge $e$ of some (new)
    chain $\gamma^*$ is crossed by an edge $e'$ of some other (new or
    old) chain. Then either $e'$ has an endpoint inside
    $\overline{R}_\gamma$, contradicting the construction, or $e'$
    crosses $\gamma$ too, to exit from $\overline{R}_\gamma$, which
    again is impossible by construction, since no edge crosses
    $\bd D_i^0$ or $\bd\U(\D_{i-1})$. This establishes \invref{i:1}.

    \invref{i:2} follows easily from the construction and from the
    preceding discussion. Note that, for each polygonal chain
    $\gamma^*$, each of its endpoints is also an endpoint of exactly
    one neighboring arc $\hat{\gamma}^*$, so the union of these arcs
    consists of closed polygonal cycles, which bound some polygonal
    region, which we call $\U_i^*$, as claimed.

    By construction, the vertical boundaries of the new polygonal
    pseudo-trapezoids of $D_i^0$ are contained in $D_i^0$ and do not
    cross any boundaries of other polygonal pseudo-trapezoids. This,
    together with \invref{i:1}, imply that the new pseudo-trapezoids
    are pairwise openly disjoint, and are also openly disjoint from
    the polygonal pseudo-trapezoids in $\T_{i-1}^*$.  It is also clear
    from the construction that each new pseudo-trapezoid
    $\sigma^* \in \T_i^* \setminus \T_{i-1}^*$ is contained in
    $D_i$. So \invref{i:3} follows.

    Finally consider \invref{i:4}.  Each new hole that is created when
    adding $D_i^0$ is of one of the following kinds:

\parpic[r]{\IncludeGraphics[page=1]{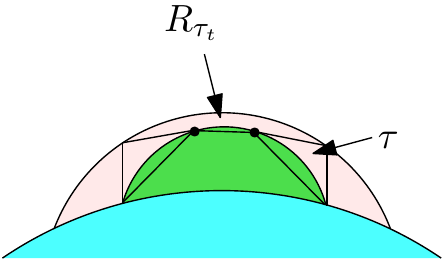}}%

    \smallskip%
    \noindent
    (a) The hole is a region of the form $\overline{R}_{\tau_b}$ or
    $\overline{R}_{\tau_t}$, for some
    $\tau \in \T_i \setminus \T_{i-1}$, such that
    $\overline{R}_{\tau_b}$ or $\overline{R}_{\tau_t}$ is contained in
    $\tau$ (if it lies outside $\tau$, it becomes part of $\tau^*$).
    Such a hole is contained in $D_i$, and is bounded by two concave
    or two convex chains, one of which, call it $\zeta^*$, is
    polygonal, and the other, $\zeta$, is part of $\bd
    D_i^0$. Moreover, $\zeta^*$, if different from the chord $e$
    connecting the endpoints of $\zeta$, passes through inner vertices
    of $\bd\U(\D_{i-1})$ that ``stick into'' the corresponding portion
    $R_{\tau_b}$ or $R_{\tau_t}$ of $\tau$; see figure on the right.

    \smallskip%
    \noindent
    (b) The hole is a region of the form $\overline{R}_\gamma$, for an
    exposed subarc $\gamma$ of an arc of $\bd\U(\D_{i-1})$, that got
    delimited by a new vertex (an endpoint of some arc of $\bd
    D_i$). These holes are similar to those of type (a).

    \smallskip%
    \noindent
    (c) The hole was part of a hole of type (a) or (b) in
    $\U(\D_{i-1})$, bounded by an arc $\gamma$ of $\bd \U(\D_{i-1})$
    and its associated polygonal chain $\gamma^*$, so that $\gamma$
    has been split into several subarcs (some hidden and some exposed)
    when adding $D_i$. For each of these subarcs $\zeta$, we construct
    an associated polygonal chain $\zeta^*$, either as a top or bottom
    side of some polygonal pseudo-trapezoid $\tau^*$ (constructed from
    a pseudo-trapezoid $\tau$ that has $\zeta$ as its top or bottom
    side), or as the polygonalization of an exposed subarc.  The
    concatenation of the chains $\zeta^*$ results in a convex
    polygonal chain that is contained in $\overline{R}_\gamma$ and
    connects the endpoints of $\gamma$. The region enclosed between
    $\gamma^*$ and $\zeta^*$ is an internal polygonal hole. Again,
    holes of type (c) can be seen all over \figref{ptrp:3}; for
    example, see the top part of $D_1$ in \figref{ptrp:3}(2').

    Holes of type (a) and (b) are \emph{boundary half-polygonal
       holes}, whereas holes of type (c) are \emph{internal polygonal
       holes}. Using the induction hypothesis that \invref{i:4} holds
    for $\U(\D_{i-1})$, we get that the union of the new holes of type
    (a) and (b), together with the old holes of type (a) and (b)
    corresponding to subarcs of $\bd \U(\D_i) \cap \bd \U(\D_{i-1})$,
    is $\U(\D_i) \setminus \U_i^*$.  This completes the proofs of
    \invref{i:1}--\invref{i:4}.
\end{proof}

\begin{theorem}%
    \thmlab{crucial}%
    (a) Let $\D$ be a collection of $m\ge 3$ planar convex pseudo-disks,
    whose union covers the plane.  Then there exists a set $V$ of
    $O(m)$ points and a triangulation $T$ of $V$ that covers the plane, such that each
    triangle $\Delta\in T$ is fully contained in some member of $\D$.

    \noindent
    (b) If $\U(\D)$ is not the entire plane, it can be partitioned into
    $O(m)$ pairwise openly disjoint triangles and caps, such that each
    triangle and cap is fully contained in some member of $\D$.
\end{theorem}

\begin{proof}
    Since the number of vertices of $M$ is $O(m)$, Euler's formula
    implies that $|E_m| = O(m)$ too. It is easily seen from the
    construction and from the invariants of \lemref{invariants}, that
    each face of $M$ is fully contained in some original pseudo-disk,
    so the same holds for each triangle. This establishes (a). Part (b)
    follows in a similar manner from the construction.
\end{proof}

\subsection{Efficient construction of the triangulation}

With some care, the proof of \thmref{crucial} can be turned into an
efficient algorithm for constructing the required triangulation.  This
is a major advantage of the new proof over the older one. The
algorithm is composed of building blocks that are variants of
well-known tools, so we only give a somewhat sketchy description
thereof%

\subsubsection{Construction of the original pseudo-trapezoids.}

(A similar approach is mentioned in \Matousek
\etal~\cite{mmpssw-ftdlm-91-conf}.)  The construction proceeds by
inserting the pseudo-disks of $\D$ in a \emph{random} order, which, for
simplicity, we denote as $D_1,\ldots,D_m$.
(Unlike the deterministic construction given above, here we
do not guarantee that each $D_i^0$ has constant complexity.
Nevertheless, as argued below, the random nature of the insertion
order guarantees that this property holds on average.)
As before, we put
$\D_i=\{D_1,\ldots,D_i\}$ for each $i$, and we maintain $\U(\D_i)$
after each insertion of a pseudo-disk. To do so efficiently, we
maintain a vertical decomposition $K_i$ of the \emph{complement} $\U_i^c$ of
the union $\U(\D_i)$ into vertical pseudo-trapezoids (as depicted in the figure on the right), and maintain,
for each $\tau\in K_i$, a \emph{conflict list}, consisting of all the
pseudo-disks $D_j$ that have not yet been inserted (i.e., with $j>i$),
and that intersect $\tau$.

\parpic[r]{\IncludeGraphics[page=1]{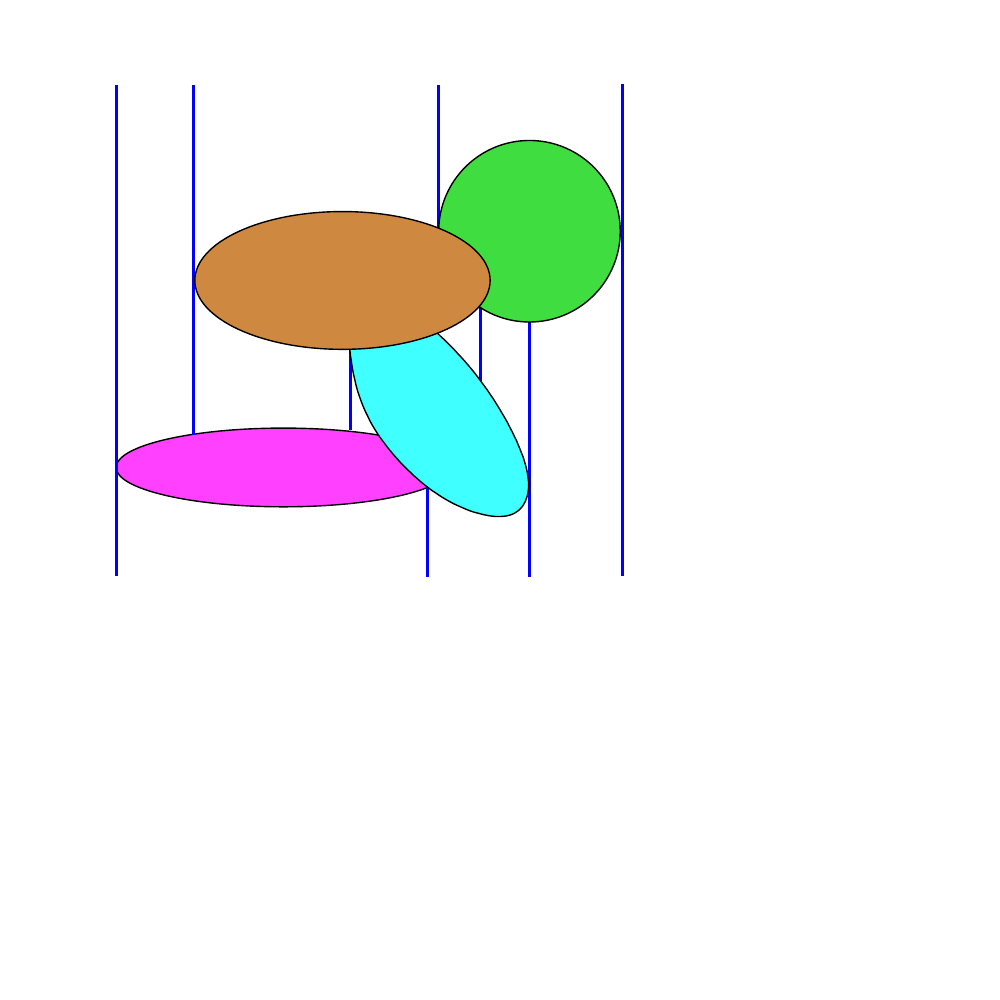}}%

\michaH{Fig. pdcompl needs fixing.}
Since the number of pseudo-trapezoids in the decomposition of the
complement of the union of any $k$ pseudo-disks
(as depicted in the figure on the right) is $O(k)$ (an easy
consequence of the linear bound on the union
complexity~\cite{klps-ujrcf-86}), a simple application of the
Clarkson-Shor technique (similar to those used to analyze many
other randomized incremental algorithms) shows that the expected overall number of these
``complementary'' pseudo-trapezoids that arise during the construction
is $O(m)$, and that the expected overall size of their conflict lists
is $O(m\log m)$.

When we insert a pseudo-disk $D_i$, we retrieve all the
pseudo-trapezoids of $K_{i-1}$ that intersect $D_i$. The union
$\bigcup_{\tau\in K_{i-1}} (D_i\cap\tau)$ is precisely $D_i^0$. For
each $\tau\in K_{i-1}$, the intersection $D_i\cap\tau$ (if nonempty) decomposes
$\tau$ into $O(1)$ sub-trapezoids (this follows from the property that
each of the four sides of $\tau$ crosses $\bd D_i$ at most twice),
some of which lie inside $D_i$ (and, as just noted, form $D_i^0$), and
some lie outside $D_i$, and form part of the new  complement of the union
$\U_i^c$.

Typically, the new pseudo-trapezoidal pieces are not necessarily real
pseudo-trapezoids, as they may contain one or two ``fake'' vertical
sides, because the feature that created such a side got ``chopped
off'' by the insertion of $D_i$, and is no longer on the
pseudo-trapezoid boundary. In this case, we ``glue'' these pieces
together, across common fake vertical sides, to form the new real
pseudo-trapezoids.  We do it both for pseudo-trapezoids that are
interior to $D_i$, and for those that are exterior.  (This gluing step
is a standard theme in randomized incremental constructions; see,
e.g., \cite{s-sfira-91}.)  This will produce (a) the desired vertical
decomposition of $D_i^0$, and (b) the vertical decomposition $K_i$ of
the new union complement $\U_i^c$. The conflict lists of the new
exterior pseudo-trapezoids (interior ones do not require conflict
lists) are assembled from the conflict lists of the pseudo-trapezoids
that have been destroyed during the insertion of $D_i$, again, in a
fully standard manner.

To recap, this procedure constructs the vertical decompositions of all
the regions $D_i^0$, so that the overall expected number of these
pseudo-trapezoids is $O(m)$, and the total expected cost of the
construction (dominated by the cost of handling the conflict lists)
is $O(m\log m)$.

\subsubsection{Construction of the polygonal chains and the
   triangulation.}
By \invref{i:2} of \lemref{invariants}, before $D_i$ was inserted,
each arc $\gamma$ of $\bd\U(\D_{i-1})$ has an associated convex
polygonal arc $\gamma^*$ with the same endpoints.  The union of the
arcs $\gamma^*$ forms a (possibly disconnected) polygonal curve within
$\U(\D_{i-1})$, which partitions it into two subsets, the
(polygonal) \emph{interior}, $\U_{i-1}^*$, which is disjoint from $\bd\U(\D_{i-1})$
(except at the endpoints of the arcs $\gamma^*$), and the (half-polygonal)
\emph{exterior}, which is simply the (pairwise openly disjoint) union of the
corresponding regions $\overline{R}_{\gamma}$.

To construct the triangulation, we maintain, for each polygonal chain
$\gamma^*$ of the boundary between the interior and the exterior, a
list of its segments, sorted in left-to-right order of their $x$-projections,
in a separate binary search tree (since the leftmost and rightmost
points of each pseudo-disk are vertices in the construction, each
chain $\gamma^*$ is indeed $x$-monotone).  We also maintain a triangulation
of the interior.  When we add $D_i$ we update the lists representing
the arcs $\gamma$ and extend the triangulation of the interior to
cover the ``newly annexed'' interior, as follows.

When $D_i$ is inserted, some of the arcs $\gamma$ of $\bd\U(\D_{i-1})$
are split into several subarcs.
At most two of these arcs
 still appear on $\bd\U(\D_i)$, and each of them is an extreme subarc of $\gamma$
(we call them, as above, \emph{exposed} arcs).  All the others are now
contained in $D_i$ (we call them \emph{hidden}). Each endpoint of any
new subarc is either an intersection point of $\bd D_i$ with
$\bd\U(\D_{i-1})$, or an endpoint of a vertical segment erected from
some other vertex of $D_i^0$. (This also includes the case where an
arc of $\bd\U(\D_{i-1})$ is fully ``swallowed'' by $D_i$ and becomes
hidden in its entirety.) In addition, $\bd\U(\D_i)$ contains
\emphi{fresh} arcs, which are subarcs of $\bd D_i$ along $\bd
D_i^0$. The fresh subarcs and the hidden subarcs form the top and
bottom sides of the new pseudo-trapezoids in the decomposition of
$D_i^0$ (where each top or bottom side may be either fresh or
hidden). To obtain the top or bottom sides of some new
pseudo-trapezoids we may have to concatenate several previously
exposed subarcs of $\bd\U(\D_{i-1})$.  These subarcs are connected at
``inner'' vertices of $\bd\U(\D_{i-1})$ which are not intersection
points of the arrangement but intersections of vertical sides of
pseudo-trapezoids which we already generated within $\U(\D_{i-1})$.)

\parpic[r]{%
   \begin{minipage}{4.8cm}%
       \IncludeGraphics{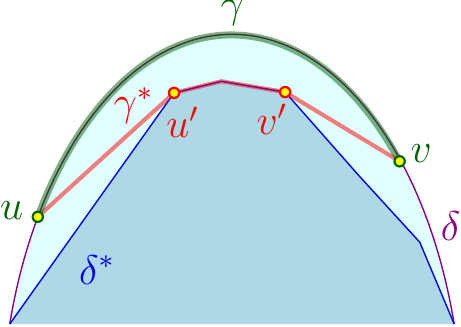}%
       \captionof{figure}{}
       \figlab{bridge}%
   \end{minipage}%
}%

The algorithm needs to construct, for each new exposed, hidden, and
fresh arc $\gamma$, its associated polygonal curve $\gamma^*$. It does
so in two stages, first handling exposed and hidden arcs, and then the
fresh ones. Let $\gamma$ be an exposed or hidden subarc, let $\delta$
denote the arc of $\bd\U(\D_{i-1})$, or the concatenation of several
such arcs, containing $\gamma$, and let $\delta^*$ be its associated
polygonal chain, or, in case of concatenation, the concatenation of
the corresponding polygonal chains.  As already noted, since the
$x$-extreme points of each pseudo-disk boundary are vertices in the
construction, $\delta$ and $\delta^*$ are both $x$-monotone.

If $\gamma=\delta$, we do nothing, as $\gamma^*=\delta^*$. Otherwise,
let $u$ and $v$ be the respective left and right endpoints of
$\gamma$. If $uv$ does not intersect $\delta^*$ then $\gamma^*$ is
just the segment $uv$. Otherwise, $\gamma^*$ is obtained from a
portion of $\delta^*$, delimited on the left by the point $u'$ of
contact of the right tangent from $u$ to $\delta^*$, and on the right
by the point $v'$ of contact of the left tangent from $v$ to
$\delta^*$, to which we append the segments $uu'$ on the left and
$v'v$ on the right. See \figref{bridge} for an illustration.

Note that the old arc $\delta$ may contain several new exposed or
hidden arcs $\gamma$, so we apply the above procedure to each such arc
$\gamma$.  After doing this, the endpoints of $\delta$ (and of
$\delta^*$) are now connected by a new convex polygonal chain
$\hat{\delta}^*$, which visits each of the new vertices along $\delta$
(the endpoints of the new arcs $\gamma$) and lies in between $\delta$
and $\delta^*$. The region between $\hat{\delta}^*$ and $\delta^*$ is
a new interior polygonal hole, and we triangulate it, e.g., into
vertical trapezoids, by a straightforward left-to-right scan.

Recall that some arcs $\tau_b$ and $\tau_t$ of new trapezoids $\tau$
may be concatenations of several hidden subarcs $\gamma_i$ (connected
at inner vertices which are not vertices of new trapezoids, as explained above).
For each such arc, say $\tau_b$, we obtain $\tau_b^*$ by
concatenating the polygonal chains $\gamma_i^*$ in $x$-monotone order.

\parpic[r]{\IncludeGraphics{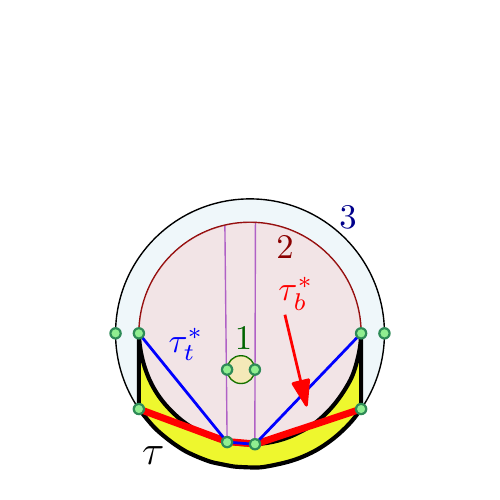}}%

We next handle the fresh arcs. Each such arc is the top or bottom side
of some new pseudo-trapezoid $\tau$, say it is the bottom side
$\tau_b$. If $\tau_t$ is also fresh, then $\tau$ is a convex
pseudo-trapezoid, and we replace each of $\tau_b$, $\tau_t$ by the
straight segment connecting its endpoints. If $\tau_t$ is hidden, we
take its associated chain $\tau_t^*$, which we have constructed in the
preceding stage, and form $\tau_b^*$ from it using the same procedure
as above: Letting $u$ and $v$ denote the endpoints of $\tau_b$, we
check whether $uv$ intersects $\tau_t^*$. If not, $\tau_b^*$ is the
segment $uv$.  Otherwise, we compute the tangents from $u$ and $v$ to
$\tau_t^*$, and form $\tau_b^*$ from the tangent segments and the
portion of $\tau_t^*$ between their contact points.
See figure on the right. We triangulate
each polygonal pseudo-trapezoid $\tau$ once we have computed
$\tau_b^*$ and $\tau_t^*$.

\subsubsection{Further implementation details.} %
The actual implementation of the construction of the polygonal chains
$\gamma^*$ proceeds as follows.  Given a new arc $\gamma$, which is a
subarc of an old arc $\delta$, we construct $\gamma^*$ from $\delta^*$
as follows.  Let $u$ and $v$ be the endpoints of $\gamma$.  We
(binary) search the list of edges of $\delta^*$ for the edge $e_u$
whose $x$-projection contains the $x$-projection of $u$ and for the
edge $e_v$ whose $x$-projection contains the $x$-projection of $v$. We
then walk along the list representing $\delta^*$ from $e_u$ towards
$e_v$ until we find the point $u'$ of contact of the right tangent
from $u$ to $\delta^*$.  We perform a similar search from $e_v$
towards $e_u$ to find $v'$. (If we have traversed the entire portion
of $\delta^*$ between $e_u$ and $e_v$ without encountering a tangent,
we conclude that $uv$ does not intersect $\delta^*$, and set
$\gamma^*:=uv$.) We extract the sublist between $u'$ and $v'$ from
$\delta^*$ by splitting $\delta^*$ at $u'$ and $v'$ and we insert the
segments $uu'$ and $vv'$ at the endpoints of this sublist to obtain
$\gamma^*$.  We create the polygonalization of fresh arcs from their
hidden counterparts in an analogous manner.  Note that we destroy the
representation of $\delta^*$ to produce the representation of
$\gamma^*$. So in case the arc $\delta$ is split into several new
subarcs, $\gamma_i$, some care has to be taken to maintain a
representation of the remaining part of $\delta^*$ after producing
each $\gamma_j^*$, from which we can produce the representation of the
remaining subarcs $\gamma_i$.

For the analysis, we note that to produce $\gamma^*$ we perform two
binary searches to find $e_u$ and $e_v$, each of which takes
$O(\log m)$ time, and then perform linear scans to locate $u'$ and
$v'$. Each edge $e$ traversed by these linear scans (except for $O(1)$
edges) drops off the boundary of the interior so we can charge this
step to $e$ and the total number of such charges is linear in the size
of the triangulation.

\subsection{The result}

\paragraph{The computation model.}
In the preceding description, we implicitly assume a convenient model
of computation, in which each primitive geometric operation that is
needed by the algorithm, and that involves only a constant number of
pseudo-disks (e.g., deciding whether two pseudo-disks or certain
subarcs thereof intersect, computing these intersection points, and
sorting them along a pseudo-disk boundary) takes constant time.
In our application, described in the next section, the pseudo-disks are
convex polygons, each having $O(k)$ edges. In this case, each primitive operation
can be implemented in $O( \log k)$ time in the standard (say, real RAM) model,
so the running time should be multiplied by this factor.

\bigskip

The preceding analysis implies the following theorem.

\begin{theorem}
    \thmlab{triangulate:union}%
    We can construct a triangulation of the union of $m$ pseudo-disks
    covering the plane, with $O(m)$ triangles, such that each triangle
    is contained in a single pseudo-disk, in $O(m \log m)$ randomized
    expected time, in a suitable model of computation where every
    primitive operation takes $O(1)$ time.
    If the union does not cover the plane, it can be decomposed into
    $O(m)$ triangles and caps, with similar properties and at the same asymptotic cost.
\end{theorem}

\begin{corollary}
    \corlab{triangulate:union}%
    Given $m$ convex polygons that are pseudo-disks, that cover that
    plane, each with at most $k$ edges, one can compute a confined
    triangulation of the plane, in $O(m \log m \log k)$ expected time.
    A statement analogous to the second part of
    \thmref{triangulate:union} holds in this case too.
\end{corollary}

\subsection{Extension to general convex shapes}

\thmref{triangulate:union} uses only peripherally the property that
the input shapes are pseudo-disks, and a simple modification (of the
analysis, not of the construction itself) allows us to extend it to
general convex shapes.  Specifically, let $\D$ be a collection of $m$
simply-shaped convex regions in the plane, such that the union
complexity of any $i$ of them is at most $u(i)$, where the complexity
is measured, as before, by the number of boundary intersection points
on the union boundary, and where $u(\cdot)$ is a monotone increasing
function satisfying $u(i) = \Omega(i)$.  We assume that the regions in
$\D$ are simple enough so that the boundaries of any pair of them
intersect only a constant number of times, and so that each primitive
operation on them can be performed in reasonable time (which we take
to be $O(1)$ in the statement below). The interesting cases are those
in which $u(i)$ is small (that is, near-linear). They include, e.g.,
the case of fat triangles, or a low-density collection of convex
regions; see \cite{abes-ibulf-14} and references therein.

Deploying the algorithm of \thmref{triangulate:union} results in the
desired confined triangulation of $\U(\D)$. Extending the analysis
to this general setup (and omitting the straightforward technical
details), we obtain the following theorem.

\begin{theorem}
    \thmlab{t:union:shapes}%
    Let $\D$ be a collection of $n$ convex shapes in the plane, such
    that the union complexity of any $i$ of them is at most $u(i)$,
    where $u(i)$ is a monotone increasing function with $u(i) = \Omega(i)$.
    Then one can compute, in $O( u(m) \log m)$ expected time, a confined
    triangulation of $\U(\D)$ with $O(u(m))$ triangles and caps (or just
    triangles if the union covers the entire plane),
    under the assumption that every primitive geometric operation
    takes $O(1)$ time.

\end{theorem}


\section{Construction of shallow cuttings %
   and approximate levels}
\seclab{approx_level}

We begin by presenting a high-level description of the technique,
filling in the technical details in subsequent subsections.  The
high-level part does not pay too much attention to the efficiency of
the construction; this is taken care of later in this section.

\subsection{Sketch of the construction}

Assume that, for a given parameter $r$, we want to approximate level
$k=n/r$ of $\Arr(H)$.  Note that when $r$ is too close to $n$, that
is, when $k$ is a constant, we can simply compute the $k$-level
explicitly and use it as its own approximation. The complexity of such
a level is $O(n)$, and it can be computed in $O(n\log n)$ time
\cite{c-dds3c-10, am-dhsrr-95} (better than what is stated in  \thmref{a:shallow:cutting} for such a large value of $r$).    We therefore assume in the remainder of
this section that $r \ll n$.

Put $k_1:=(1+c)n/r$ and $k_2:=(1+2c)n/r$, for a
suitable sufficiently small (but otherwise arbitrary) constant
fraction $c$.  The analysis of Clarkson and Shor~\cite{cs-arscg-89}
implies that the overall complexity of $L_{\le k_2}$ (the first $k_2$
levels of $\Arr(H)$) is $O(nk^2)$.  This in turn implies that there
exists an index $k_1\le \xi\le k_2$ for which the complexity $|L_\xi|$ of
$L_\xi$ is $O(nk^2/(cn/r)) = O(nk/c) = O(n^2/(cr))$.  We fix such a
level $\xi$, and continue the construction with respect to $L_\xi$
(slightly deviating from the originally prescribed value of $k$).
However, to simplify the notation for the current part of the analysis, we use
$k$ to denote the nearby level $\xi$, and will only later return to the
original value of $k$.

The next step is to decompose the $xy$-projection of the $k$-level
$L_k$ into a small number of connected polygons, from which the
approximate level will be constructed. We first review the existing
machinery, already mentioned in the introduction, for this step.

\paragraph{Decomposing a level into a small number of polygons.}%

Let $H$, $k$, and $L_k$ be as above.
It is convenient to assume that the faces of $L_k$ are
triangles; this can be achieved by triangulating each face, without
affecting the asymptotic complexity of $L_k$. In particular, the
$k$-level (or, rather, its $xy$-projection)
can then be interpreted as a planar, triangulated and biconnected
graph (a graph is \emph{biconnected}, if any pair of vertices are
connected by at least two vertex-disjoint paths).

As has been discovered over the years, planar graphs can be efficiently decomposed
into smaller pieces that are well behaved. This goes back to the planar separator
theorem of Lipton and Tarjan~\cite{lt-stpg-79}, Miller's cycle separator theorem
\cite{m-fsscs-86}, and Frederickson's divisions \cite{f-faspp-87}, and has eventually
culminated in the fast $\kappa$-division algorithm of Klein \etal~\cite{kms-srsdp-13}.
Specifically, for a (specific drawing of a) planar triangulated and biconnected graph $G$
with $N$ vertices, and for a parameter $\kappa<n$, a \emph{$\kappa$-division}
of $G$ is a decomposition of $G$ into several connected subgraphs
$G_1,\ldots,G_m$, such that (i) $m=O(N/\kappa)$; (ii) each $G_i$ has at
most $\kappa$ vertices; (iii) each $G_i$ has at most $\beta\sqrt{\kappa}$
\emph{boundary vertices}, for some absolute constant $\beta$, namely, vertices that belong to at least one
additional subgraph; and (iv) each $G_i$ has at most $O(1)$
\emph{holes}, namely, faces of the induced drawing of $G_i$ that are
not faces of $G$ (as they contain additional edges and vertices of
$G$).  \smallskip%
Such a division can be computed in $O(N)$ time \cite{kms-srsdp-13}.\footnote{%
  The algorithm of \cite{kms-srsdp-13} in fact constructs $\kappa$-divisions for
  a geometrically increasing sequence of values of the parameter $\kappa$,
  in overall $O(N)$ time. }

\paragraph{The construction, continued.}

We set
\begin{align*}
    t:= \frac{cn-43.5r}{9\beta r} ,
\end{align*}
where $\beta $ is the constant from property (iii) of the \emph{$\kappa$-division} given above.
Notice that since $r\ll n$ we have $t > 1$.
Let $L'_k$ be the projection of $L_k$ to the $xy$-plane. We turn
$L'_k$ into a triangulated and biconnected planar graph $G'_k$ by adding a new vertex $v_\infty$,
replacing each ray $[p,\infty)$ by the edge $(p,v_\infty)$, and triangulating the
newly created faces.
We apply the planar subdivision
algorithm of \cite{kms-srsdp-13, f-faspp-87}, as just reviewed, and
construct a $t^2$-division of $G'_k$. This subdivision produces
\begin{align*}
    m := O(|L_k|/t^2) = O\left(\frac{n^2/(cr)}{c^2n^2/r^2}\right) =
    O(r/c^3)
\end{align*}
connected polygons, $P'_1,\ldots,P'_m$, with pairwise disjoint
interiors. Each polygon $P'_i$ corresponds to
a (possibly unbounded) polygon $P_i$ of $L'_k$.
The union of $P_1,\ldots,P_m$
 covers the entire $xy$-plane, and its edges
are projections of (some) edges of $L_k$ (including diagonals drawn to
triangulate the original faces of $L_k$).

By construction, each $P_i$ is connected and has at most $\beta t$ edges
(and also contains $O(t^2)$ edges and vertices of the projected $k$-level in its
interior).
Let $C_i$ denote the convex hull of $P_i$, for
$i=1,\ldots,m$. As we show in \corref{mchpd} in \secref{p-disks}
below, $\C:=\{C_1,\ldots,C_m\}$ is a collection of $m$ (possibly
unbounded) convex \emph{pseudo-disks} whose union is the entire plane.

We then apply \thmref{crucial} to $\C$ and obtain a set $S$ of $O(m)$
points in the $xy$-plane, and a triangulation $T$ of $S$, such that
each triangle $\Delta\in T$ is fully contained in some hull $C_i$ in
$\C$. For a point $p$ in the $xy$-plane, we denote by $\Lift_k(p)$ the
\emph{lifting} of $p$ to the $k$-level, i.e., the point on the level
that is co-vertical with $p$.

For a bounded
triangle $\Delta$, $\Lift_k(\Delta)$ is defined as the triangle spanned by the
lifted images of the three vertices of $\Delta$; we lift an unbounded
triangle $\Delta$ with vertices $p$ and $q$ by lifting $pq$ to
$\Lift_k(p)\Lift_k(q)$ as before, and lifting each of its rays, say
$[p,\infty)$, to a ray $\Lift([p,\infty))$ emanating from $\Lift_k(p)$
in a direction parallel to the plane which is vertically above
$[p,\infty)$ at infinity.  If the liftings $\Lift([p,\infty))$, and
$\Lift([q,\infty))$, and the edge $\Lift_k(p)\Lift_k(q)$ are not on
the same plane, we add another ray, say $r$, emanating from $p$
parallel to $[q,\infty)$. We add to $T'$ the unbounded triangle
spanned by $\Lift([q,\infty))$, $\Lift_k(p)\Lift_k(q)$, and $r$, and
the unbounded wedge spanner by $r$ and $\Lift([p,\infty))$.
Let $T'$ denote the corresponding collection of triangles in $\Re^3$,
given by $T' = \{ \Lift_k(\Delta) \mid \Delta\in T \}$.

Note that the triangles of $T'$ are in general not
contained in $L_k$. However, for each triangle $\Delta'\in T'$, its
vertices lie on $L_k$, and, as we show in \lemref{xtri} below, at most
$9\beta t + 43.5$ planes of $H$ can cross $\Delta'$. This implies, returning
now to the original value of $k$, that $\Delta'$
fully lies between the levels
\begin{align*}
    \xi\pm(9\beta t + 43.5) = \xi\pm cn/r
\end{align*}
of $\Arr(H)$. In particular, $\Delta'$ lies fully above the level
\begin{align*}
    \xi-cn/r \ge k_1-cn/r = n/r=k ,
\end{align*}
and fully below the level
\begin{align*}
    \xi+cn/r \le k_2+cn/r = (1+3c)n/r = (1+3c)k .
\end{align*}
The lifted triangulation $T'$ forms a polyhedral terrain that
consists of $O(r/c^3)$ triangles and is contained between the levels
$k=n/r$ and $(1+3c)k$. That is, for a given $\eps>0$, choosing
$c=\eps/3$ makes $T'$ an $\eps$-approximation of $L_k$, and
we obtain the following result.
\begin{theorem}%
    \thmlab{apxlev}%
    Let $H$ be a set of $n$ non-vertical planes in $\Re^3$ in general
    position, and let $r\le n$, $\eps>0$ be given parameters. Then there exists
    a polyhedral terrain consisting of $O(r/\eps^3)$ triangles, that
    is fully contained between the levels $n/r$ and $(1+\eps)n/r$ of
    $\Arr(H)$.
\end{theorem}

To turn this approximate level into a shallow cutting, replace each
$\Delta'\in T'$ (including each of the unbounded triangles as
constructed above) by the semi-unbounded vertical prism $\Delta^*$
consisting of all the points that lie vertically below $\Delta'$.
This yields a collection $\Xi$ of prisms, with pairwise disjoint
interiors, whose union covers $L_{\le n/r}$, so that, for each prism
$\tau$ of $\Xi$, we have (a) each vertex of $\tau$ lies at level
(at least $k$ and) at most $(1+\frac23 \eps)k$, and, as will be
established in \lemref{xtri} below, (b) the top triangle of
$\tau$ is crossed by at most $\frac13 \eps k$ planes of $H$ (in the preceding analysis, we wrote this bound as
$9\beta t + 43.5 = \frac{cn}{r}$; this is the same value, recalling that $\eps=3c$ and $k=n/r$). Hence,
as is easily seen, each
prism of $\Xi$ is crossed by at most $(1+\eps)n/r$ planes,
so $\Xi$ is the desired shallow cutting. That is, we have the following result.
\begin{theorem}%
    \thmlab{shallow-cutting}%
    Let $H$ be a set of $n$ non-vertical planes in $\Re^3$ in general
    position, let $k<n$ and $\eps>0$ be given parameters, and put $r=n/k$.
    Then there exists a $k$-shallow $((1+\eps)/r)$-cutting of $\Arr(H)$,
    consisting of $O(r/\eps^3)$ vertical prisms (unbounded from below).
    The top of each prism is a triangle that is fully contained between the
    levels $k$ and $(1+\eps)k$ of $\Arr(H)$, and these triangles
    form a polyhedral terrain (we say that such a terrain approximates the $k$-level
    $L_k$ up to a relative error of $\eps$).
\end{theorem}

\subsection{Efficient implementation}%
\seclab{sec:efficient}

We next turn our constructive proof into an efficient algorithm, and show:

\begin{theorem}
    \thmlab{a:shallow:cutting}%
Let $H$ be a set of $n$ non-vertical planes in $\Re^3$ in general
    position, let $k<n$ and $\eps>0$ be given parameters, and put $r=n/k$.
   One can construct the $k$-shallow $((1+\eps)/r)$-cutting of $\Arr(H)$ given in
    \thmref{shallow-cutting}, or, equivalently, the
    $\eps$-approximating terrain of the $k$-level in
    \thmref{apxlev}, in $O(n + r \eps^{-6} \log^3 r)$
    expected time. Computing the conflict lists of the vertical prisms
    takes an additional $O(n(\eps^{-3}+ \log\frac{r}{\eps}))$ expected time.
    The algorithm, not including the construction of the conflict lists, computes a correct $\eps$-approximating terrain with
probability at least $1-1/r^{O(1)}$. If we also compute the conflict lists then we can verify, in $O(n/\eps^3)$ time, that the cutting is indeed correct and thereby
make the algorithm always succeed, at the cost of increasing its expected running time by a constant factor.
\end{theorem}
\begin{proof}
 Let $(H,\R)$ denote the range space where
    each range in $\R$ corresponds to some vertical segment or ray
    $e$, and is equal to the subset of the planes of $H$ that are
    crossed by $e$.  Clearly, $(H,\R)$ has finite VC-dimension (see,
    e.g.,~\cite{sa-dsstg-95}).  We draw a random sample $S$ of
    ${\displaystyle n' = \frac{br}{\eps^2}\log r}$ planes from $H$,
    where $b$ is a suitable constant.  For $b$ sufficiently large,
    such a sample is a \emph{relative
       $\left(\frac{1}{r},\eps \right)$-approximation} for $(H,\R)$,
    with probability $\geq 1-1/r^{O(1)}$; see \cite{hs-rag-11} for the
    definition and properties of relative approximations.
    This means that each vertical segment or ray that intersects $x\ge n/r$ planes of $H$
    intersects between $({1+\eps})\frac{n'}{n} x$ and
    $({1-\eps})\frac{n'}{n} x$ planes of $S$, and each vertical
    segment or ray that intersects $x < n/r$ planes of $H$ intersects at most
    $\frac{n'}{n}x + \eps\frac{n'}{r}$ planes of $S$. (This holds,
    with probability $\geq 1-1/r^{O(1)}$, for all vertical segments and rays.)

    The strategy is to use (the smaller) $S$ instead of $H$ in the
    construction, as summarized in \thmref{shallow-cutting}, and then
    argue that a suitable approximate level in $\Arr(S)$ is also an
    approximation to level $k$ in $\Arr(H)$ with the desired
    properties.  Set
    \begin{align*}
      k' = \frac{b(1+\eps)}{\eps^2}\log r ,\quad\text{and}\quad
      t' = \frac{b(1+\eps)}{\eps}\log r = \eps k' .
    \end{align*}
    We choose a random index $\xi$ in
    the range $[k',k'+t']$, construct the $\xi$-level of $\Arr(S)$,
    and then apply the construction of the proof of \thmref{shallow-cutting} to this level, as will be detailed below.

    Before doing this, we first show that the $\xi$-level in $\Arr(S)$ is a good
    approximation to level $k$ in $\Arr(H)$.
Consider a point $p$ on level $k$ of $\Arr(H)$. By the property specified above of a relative
$\left(\frac{1}{r},\eps \right)$-approximation, it follows that the level of $p$ in
$\Arr(S)$ is at most $(1+\eps)(n'/n)(n/r) = k'$. Similarly, let  $p$ be a point at level larger than, say, $(1+4\eps)k$
of $\Arr(H)$. Then the level of $p$ in $\Arr(S)$ is at least $(1-\eps)(n'/n)(1+4\eps)(n/r)\ge (1+\eps)k'= k'+t'$, for $\eps \le 1/2$.
Since this holds
    with probability $\geq 1-1/r^{O(1)}$, for every point $p$, we
    conclude that the entire $\xi$-level of $\Arr(S)$ is between levels $k$ and $(1+4\eps)k$
    of $\Arr(H)$, with probability $\geq 1-1/r^{O(1)}$.

    We can now apply the machinery in \thmref{shallow-cutting}.  The
    first step in this analysis is to construct the $\xi$-level in
    $\Arr(S)$.  Rather than just constructing that level, we compute
    all the first $k'+t'$ levels, using a randomized algorithm of Chan
    \cite{c-rshrr-00},\footnote{
       The paper of Chan \cite{c-rshrr-00} does not use shallow
       cuttings, so we are not using ``circular reasoning'' in
       applying his algorithm.} which takes
    \begin{align*}
        O(n' \log n' + n' (k')^2 )%
        &=%
        O\pth{ \frac{r\log r}{\eps^2}\pth{ \log \frac{r}{\eps} +
              \frac{\log^2 r}{\eps^4} } }%
        =%
        O\pth{r \eps^{-6} \log^3 r}
    \end{align*}
    expected time. We can then easily extract the desired (random)
    level $\xi$. In expectation (over the random choice of $\xi$), the
    complexity of the $\xi$-level is
    \begin{align*}
    n_1 = O\pth{ n' (k')^2 / (\eps k') } = O\pth{ n'k' / \eps }
    = O\left( \frac{r}{\eps^5} \log^2 r \right) ,
    \end{align*}
    and we assume in what follows that this is indeed the case.

    We now continue the implementation of the construction in a straightforward manner.
    We already have the random $\xi$-level. We project it onto the
    $xy$-plane, and construct a $(t')^2$-division of the projection,
    in $O(n_1)$ time. It consists of
    \begin{align*}
        m := O(n_1/(t')^2) = O\left( \frac{n'}{\eps^3 k'} \right) = O(r/\eps^3)
    \end{align*}
    pieces, each with $O(t') = O\left(\frac{1}{\eps} \log r \right)$ edges.
    We compute their convex hulls in $O(mt')=O(\frac{r}{\eps^4}\log r)$ time, and then construct the corresponding confined
    triangulation, in overall time
    \begin{align*}
   O(mt')+ O(m\log m\log t') = O\left( \frac{r}{\eps^4}\log r+ \frac{r}{\eps^3} \log \frac{r}{\eps}
        \log\left( \frac{1}{\eps} \log r\right) \right) .
    \end{align*}
    Finally, we need to lift the vertices of the resulting
    triangles to the $\xi$-level of $\Arr(S)$. This can be done,
using a point location data structure over the $xy$-projection of this level, in
$O(n_1\log n_1 + m \log n_1) = O(\frac{r}{\eps^5}\log^2 r \log\frac{r}{\eps})$.
    We obtain a terrain $T'$, with the claimed number of triangles, which is
    an $\eps$-approximation of the $k'$-level of $\Arr(S)$, and which lies
    above that level; the last two properties follow from \thmref{shallow-cutting}, applied to $\Arr(S)$  with suitable parameters. That is, the level in $\Arr(S)$ of each point on $T'$ is between $k'$ and
$(1+\eps)(k'+t') = (1+\eps)^2 k' < (1+3\eps)k'$ (for $\eps < 1$).
We now repeat the preceding analysis, with $3\eps$ replacing $\eps$, and conclude that $T'$ lies fully between level $k$ and
level $(1+12\eps)k$ of $\Arr(H)$. A suitable scaling of $\eps$ gives us the desired approximation in
$\Arr(H)$.

    This at last completes the construction (excluding the construction of the conflict lists).
    Its overall expected cost is $O( n + r\eps^{-6} \log^3 r )$.

    To complete the construction, we next turn to its final stage which is to compute, for every semi-unbounded vertical prism $\Delta^*$ stretching below a triangle  $\Delta'\in T'$,
    the set of planes of $H$ that intersect it (i.e., the conflict
    list of the prism). To this end, we put the vertices of $T'$ into the range
    reporting data-structure of Chan \cite{c-rshrr-00} ---
    specifically, after preprocessing, in $O(\frac{r}{\eps^3}\log\frac{r}{\eps})$ expected time,
    given a query half-space $h^+$, one can report the points in
    $h^+ \cap T'$ in $O( \log \frac{r}{\eps} + |h^+ \cap T'|)$ expected
    time (we recall again that this data range reporting structure of Chan is simple and does not use shallow cutting).
We query this data structure with
the set of halfspaces $h^+$, bounded from below by the respective
 planes $h\in H$, and, for each vertex $x$ of $T'$
that we report, we add $h$ to the conflict lists of the prisms incident to $x$.
This takes $O( n \log \frac{r}{\eps} + \frac{n}{\eps^3})$ expected time, since the total size of the conflict lists is $O(\frac{r}{\eps^3}\cdot \frac{n}{r})=O(\frac{n}{\eps^3})$ (in expectation and with probability $\geq 1-1/r^{O(1)}$).

Recall that the probability that the  sample $S$ fails to be a
relative $\left(\frac{1}{r},\eps \right)$-approximation for $(H,\R)$
is at most $1/r^{O(1)}$. When this happens, $T'$ may fail to be
the desired $k$-shallow $((1+\eps)/r)$-cutting. Such a failure
happens if and only if there exists a vertex of $T'$ whose conflict list is of size
smaller than $k$ or larger than $(1+12\eps) k$.
When we detect such a conflict list, we repeat the entire computation.
Since the failure probability is small the expected number of times we will repeat the computation is (a small) constant.
\end{proof}

We now proceed to fill in the details of the various steps of the
construction.

\subsection{The convex hulls of pairwise openly %
   disjoint polygons are pseudo-disks}%
\seclab{p-disks}

\begin{lemma}%
    \lemlab{chpd}%
    Let $P$ and $P'$ be two connected polygons in the plane with
    disjoint interiors, and let $C$ and $C'$ denote their respective
    convex hulls. Then $\bd C$ and $\bd C'$ cross each other at most
    twice.
\end{lemma}%

\begin{proof}
    For simplicity of exposition, we assume that $P$ and $P'$ are in
    general position, in a sense that will become more concrete from
    the proof.  It is easily argued that this can be made without loss
    of generality.

    \parpic[r]{%
       \IncludeGraphics{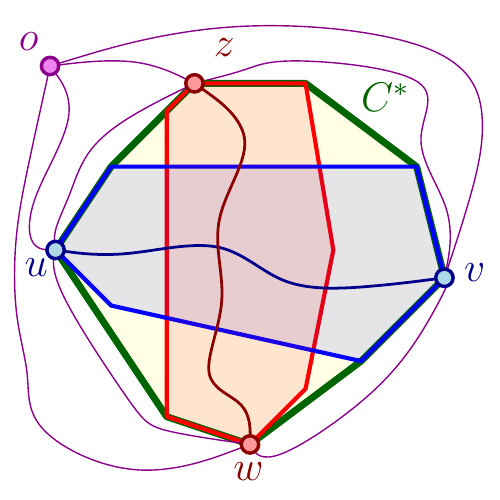}%
    }%

    Assume, for the sake of contradiction, that $\bd C$ and $\bd C'$
    cross more than twice (in general position, the boundaries do not
    overlap). This implies that each of $\bd C\setminus C'$,
    $\bd C'\setminus C$ is disconnected, and thus there exist four
    vertices $u,w,v$, and $z$ of the boundary of
    $C^* =\CHX{C\cup C'}$, that appear along $\bd C^*$ in this
    circular order, so that $u,v\in \bd C\setminus C'$ and
    $w,z\in \bd C'\setminus C$.  Clearly, $u$ and $v$ are also
    vertices of $P$, and $w$ and $z$ are vertices of $P'$.

    We show that this scenario leads to an impossible planar drawing
    of $K_5$.  For this, let $o$ be an arbitrary point outside
    $C^*$. Connect $o$ to each of $u,v,w,z$ by noncrossing arcs that
    lie outside $C^*$, and connect $u,w,v$, and $z$ by the four
    respective portions of $\bd C^*$ between them.  Finally, connect
    $u$ to $v$ by a path contained in $P$, and connect $w$ to $z$ by a
    path contained in $P'$. The resulting ten edges are pairwise
    noncrossing, where, for the last pair of edges, the property
    follows from the disjointness of (the interiors of) $P$ and $P'$.
    The contradiction resulting from this impossible planar drawing of
    $K_5$ establishes the claim.
\end{proof}

\noindent
Note that the above proof does not require the polygons to be simply
connected.

\begin{corollary}%
    \corlab{mchpd}%
    Let $\P = \{ P_1, \ldots, P_m\}$ be a set of $m$ pairwise openly
    disjoint connected polygons in the plane, and let $C_i$ denote the
    convex hull of $P_i$, for $i=1,\ldots,m$. Then
    $\C := \{ C_1,\ldots, C_m\}$ is a collection of $m$ convex
    pseudo-disks.
\end{corollary}


\subsection{Crossing properties of the planar subdivision}

Recall that our construction computes a $t^2$-division of the $xy$-projection $L'_k$ of $L_k$ where
$t:=(cn/r-43.5r)/9\beta r$
 (recall that $k=n/r$, $r \ll n$, and $\beta$ is a constant). Our goal in the rest
of this section is to show that the lifting  $\Lift_k(\Delta)$ of any triangle $\Delta$ contained in the
convex hull $C$ of a subgraph (``piece'') $P$ of this decomposition intersects
at most $ck$ planes of $H$.
We prove this for bounded triangles, the proof for unbounded triangles is similar.

Recall that for a point $p$ in the $xy$-plane, we denote by
$\Lift_k(p)$ the (unique) point that lies on $L_k$ and is co-vertical
with $p$.
The \emph{crossing distance} $\crX(p,q)$ between any pair of points
$p,q\in\Re^3$, with respect to $H$, is the number of planes of $H$
that intersect the closed segment $pq$.  The crossing distance is a
quasi-metric, in that it is symmetric and satisfies the triangle
inequality.  For a connected set $X \subseteq \Re^3$, the
\emph{crossing number} $\crX(X)$ of $X$ is the number of planes of $H$
intersecting $X$ (thus $\crX(p,q)$ is the crossing number of the
closed segment $pq$).

\begin{lemma}%
    \lemlab{middle}%
    Let $p, q, r$ be three collinear points in the $xy$-plane, such
    that $q\in pr$, and let $p' = \Lift_k(p)$, $q' = \Lift_k(q)$, and
    $r' = \Lift_k(r)$; these points, that lie on the $k$-level, are in
    general not collinear.  Let $q''$ be the intersection of the
    vertical line through $q$ with the segment $p' r'$. Then we have
    $\crX(q'',q') \leq \frac{1}{2} \crX(p',r') + 8.5$.
\end{lemma}
\begin{proof}

For a point $q$ we denote by $\level(q)$ the number of planes lying vertically strictly below $q$.
Put $k'' = \level(q'')$. The point $p'$ lies at level $k$, which is the closure of all points of level $k$.
So the number of planes lying vertically strictly below $q$ is $k$ if $p'$ is in the relative interior of a face of level $k$, at least $k-1$ if $p'$ is
in the relative interior of an edge of level $k$, and at least  $k-2$ if $p'$ is a vertex in level $k$. In either case we have $\level(p') \ge k-2$, and similarly for $r'$, and thus
$$
  \crX(p',q'')  \ge |\level(p')-\level(q'')| \ge |k-k''|-2\ ,
$$
and
$$
  \crX(q'',r')  \ge |\level(p')-\level(q'')| \ge |k-k''| -2\ .
$$
\parpic[r]{\IncludeGraphics[page=1]{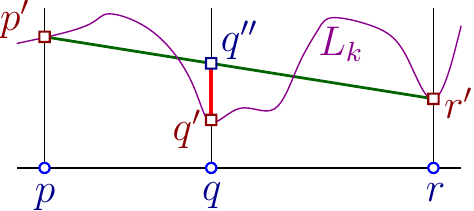}}
On the other hand we have
$$
  \crX(q',q'') \le |k'' - \level(q')| + 3 \le |k-k''| + 5\ .
$$
(Indeed, if $q''$ lies above $q'$ then $|k''-\level(q')| \le |k''-(k-2)| \le |k''-k|+2$,
and if $q'$ lies above $q''$ then $|k''-\level(q')| \le |k''-k|$. In addition, the
 difference in the levels of $q'$ and $q''$ does not count the at most three planes that intersect
$q'q''$ at $q''$ if $q''$ is above $q'$ and at $q'$ otherwise).
Hence,
\begin{align*}
  \crX(q',q'')& \le \frac{1}{2}\left( \crX(p',q'') +  \crX(q'',r') + 4\right) + 5 \\
& \le  \frac{1}{2}(\crX(p',r')+3) + 7 = \frac{1}{2}\crX(p',r') + 8.5.
\end{align*}
\end{proof}

\medskip%
In what follows, we consider polygonal regions contained in $L_k$,
where each such region $R$ is a connected union of some of the faces
of $L_k$. The $xy$-projection of $R$ is a connected polygon in the
$xy$-plane, and, for simplicity, we refer to $R$ itself also as a
polygon.

\smallskip%
\noindent%
\begin{minipage}{0.7\linewidth}%
    \begin{lemma}%
        \lemlab{seg:len}%
        Let $H$ be a set of $n$ non-vertical planes in $\Re^3$ in
        general position.  Let $P'$ be a bounded connected polygon
        with $t$ edges that lies on the $k$-level $L_k$ of $\Arr(H)$,
        such that all the boundary edges of $P'$ are edges of $L_k$.
        Let $p'$ be a vertex of the external boundary of $P'$, and let
        $q$ be any point in the convex hull $C$ of the $xy$-projection
        $P$ of $P'$. Then the crossing distance between $p'$ and
        $q' = \Lift_k(q)$ is at most $3t+14.5$.
    \end{lemma}
\end{minipage}%
\begin{minipage}{0.3\linewidth}
    \centerline{
       \IncludeGraphics{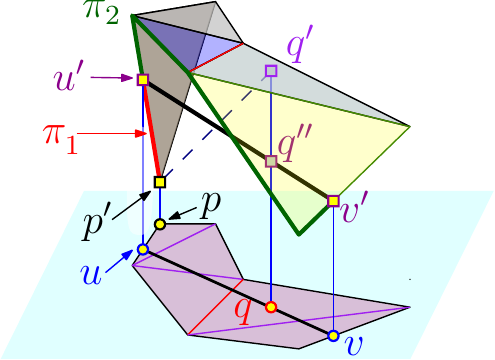}%
    }
\end{minipage}

\begin{proof}
    Since $q$ lies in $C$, we can find two points $u$, $v$ on the
    external boundary of $P$ such that $q \in uv$. Put
    $q' = \Lift_k(q)$, $u' = \Lift_k(u)$, and $v' = \Lift_k(v)$, and
    denote by $q''$ the point that lies on the segment $u'v'$ and is
    co-vertical with $q$.  We have
    \begin{align*}
      \crX(p',q')%
      &\leq%
        \crX(p',u') + \crX(u',q'') + \crX(q'',q')%
        \leq%
        \crX(p', u') + \crX(u',v') + \crX(q'',q') .
    \end{align*}
    Let $\pi_1$ and $\pi_2$ be the two portions of the external
    boundary that connect $p'$ and $u'$, and $u'$ and $v'$,
    respectively, and that
    do not overlap.  Now, by \lemref{middle}, we have
    \begin{math}
        \crX(q'',q') \le \frac{1}{2} \crX(u',v') + 8.5 ,
    \end{math}
    so we get
    \begin{align*}
        \crX(p',q')%
        &\leq%
        \crX(p',u') + \frac{3}{2}\crX(u',v') + 8.5 %
        \leq %
        \crX(\pi_1) + \frac{3}{2}\crX(\pi_2) + 8.5%
        \leq \frac{3}{2} \crX(\bd P') + 13,
    \end{align*}
    where $\bd P'$ denotes the external boundary of $P'$, and where the last inequality follows because $\frac{3}{2}\crX(\pi_1) + \frac{3}{2}\crX(\pi_2)$ double counts the planes that pass through $u'$, adding at most $\frac{3}{2} \cdot 3 = 4.5$ to the bound.

    To bound the number of planes of $H$ that intersect $\bd P'$,
    consider its vertices $p_1, p_2, \ldots, p_{t}$ (the actual number
    of vertices might be smaller since $P'$ may not be simply connected). Observe that
    $p_1$ is contained in three planes. For each $i$, $p_i$ lies on at
    most two planes that do not contain $p_{i-1}$ (there are two such
    planes when $p_{i-1} p_i$ is a diagonal of an original face of the
    untriangulated level $L_k$).  Furthermore, the open segment
    $p_{i-1} p_i$ does not cross any plane, and each plane that contains it contains both its endpoints. Therefore, the number $\crX(\bd P')$ of
    planes of $H$ that intersect $\bd P'$ satisfies
    \begin{math}
        \crX(\bd P') \leq 3 + 2(t-1) = 2t +1,
    \end{math}
    from which the lemma follows.
    (Note that this analysis is somewhat conservative---for example,
    if the polygon $P'$ uses only original edges of the $k$-level, the
    bound drops to  $t+2$.)~~
\end{proof}%

\begin{lemma}%
    \lemlab{xtri}%
    Let $H$ be a set of $n$ non-vertical planes in $\Re^3$ in general position, and
    let $P'$ be a connected polygon with $t$ edges, such that $P'$
    lies on the $k$-level $L_k$ of $\Arr(H)$,
    and such that all the boundary edges of $P'$ are edges of $L_k$.
    Then, for any triangle
    $\Delta=\Delta pqr$ that is fully contained in the convex hull of
    the $xy$-projection of $P'$, the number $\crX(\Delta')$ of planes
    of $H$ that cross the triangle $\Delta'=\Delta p'q'r'$, where
    $p'=\Lift_k(p)$, $q'=\Lift_k(q)$, $r'=\Lift_k(r)$, is at most
    $9t+43.5$.
\end{lemma}

\begin{proof}
    Let $w$ be any vertex of the external
    boundary of $P'$. Any plane that crosses $\Delta'$ must also cross
    two of its sides. Moreover, by \lemref{seg:len} and the triangle
    inequality,
    \begin{align*}
        \crX(p',q') \le \crX(w,p') + \crX(w,q') \le 2(3t+14.5) ,
    \end{align*}
    and similarly for $\crX(p',r')$ and $\crX(q',r')$. Adding up these
    bounds and dividing by $2$, implies the claim.~~
\end{proof}%



\section{Applications}
\seclab{applications}

\subsection{Constructing layered cuttings of %
   the whole arrangement}
\seclab{layered:cutting}

In our first application, we extend \Matousek's construction~\cite{m-cen-90}
of cuttings in planar arrangements to the three-dimensional case.
That is, we apply our technique to construct, for a set $H$ of $n$ planes
in $\Re^3$, a \emph{layered cutting} of the whole arrangement $\Arr(H)$,
of optimal size $O(r^3)$. Rather informally (precise statements and full
details are given below), for a given parameter $r<n$, we partition
$\Arr(H)$ into $\Theta(r)$ layers, as follows. We choose a suitable sequence
of $\Theta(r)$ levels, roughly $n/r$ apart, and approximate each level in the
sequence, as above. Then, for each pair of consecutive approximate levels,
we triangulate the layer between them into vertical triangular prisms,
each straddling the layer from its bottom level to its top level.
The actual construction is slightly more involved, and the analysis is
considerably more complicated than the one for the planar case in~\cite{m-cen-90}.

\subsubsection{Preliminaries.}

To construct such a layered cutting, we need the following technical tools.

\begin{lemma}%
    \lemlab{xedges}%
    Let $H$ be a set of $n$ non-vertical planes in $\Re^3$ in general
    position.  The number of pairs of edges $(e,e')$ of $\Arr(H)$ such
    that the $xy$-projections of $e$ and of $e'$ cross each other, and
    the unique vertical segment connecting $e$ and $e'$ does not cross any
    other plane of $H$, is $O(n^3)$.
\end{lemma}

\begin{proof}
    The number of such pairs of edges is at most
    $\sum_{c\in\Arr(H)} |c|^2$, where the sum ranges over all
    three-dimensional cells $c$ of $\Arr(H)$, and where $|c|$ denotes
    the overall complexity of $c$. This latter sum is known to be
    $O(n^3)$---it is an easy consequence of the Zone Theorem in three
    dimensions; see Aronov \etal~\cite{ams-sscch-94}.
\end{proof}%

\begin{lemma}%
    \lemlab{x:t:edges}%
    Let $H$ be a set of $n$ non-vertical planes in $\Re^3$ in general
    position, and let $q$ be a parameter.  The number of pairs of
    edges $(e,e')$ of $\Arr(H)$ such that the $xy$-projections of $e$
    and of $e'$ cross each other, and the unique vertical segment
    connecting $e$ and $e'$ crosses at most $q$ planes of $H$, is
    $O(n^3q)$.
\end{lemma}

\begin{proof}
    This follows by a standard application of the Clarkson-Shor
    technique~\cite{cs-arscg-89} to the bound stated in
    \lemref{xedges}: The number of planes defining a pair $(e,e')$
    is four, and the Clarkson-Shor analysis then yields the bound
    $O(q^4(n/q)^3) = O(n^3q)$.
\end{proof}%

\subsubsection{Constructing a layered cutting of $\Arr(H)$.}
Let $H$ be a set of $n$ non-vertical planes in $\Re^3$ in general
position, and let $r<n$ be a parameter.  Our goal is to construct a
$(1/r)$-cutting of the entire $\Arr(H)$, of optimal size $O(r^3)$. To
do so, consider some fixed sequence of $2r$ levels
\begin{align*}
 k_2^- < k_1^+ < k_3^- < k_2^+ < \cdots < k_r^- < k_{r-1}^+  ,
\end{align*}
where each pair of consecutive indices in this sequence are at
distance at least $n/(4r)$.  That is, we form a sequence of
overlapping intervals $[-\infty,k_1^+],\ldots,[k_r^-,\infty]$, so that
each interval starts after the preceding one starts and before it
ends, and no three intervals share a common index.
We choose such a sequence in the following random manner.
Fix the intervals
\begin{align*}
    I_i^- & = [(i-3/2)n/r+1,(i-5/4)n/r], && \text{for $i=2,\ldots,r$} \\
    I_i^+ & = [in/r+1,(i+1/4)n/r] ,  &&\text{for $i=1,\ldots,r-1$} \ .
\end{align*}
Then choose $k_i^-$ (resp., $k_i^+$) uniformly at random from $I_i^-$
(resp., $I_i^+$), for $i=1,\ldots,r$.  See \figref{intervals}.

\begin{figure*}
    \begin{center}
        \IncludeGraphics[page=1]{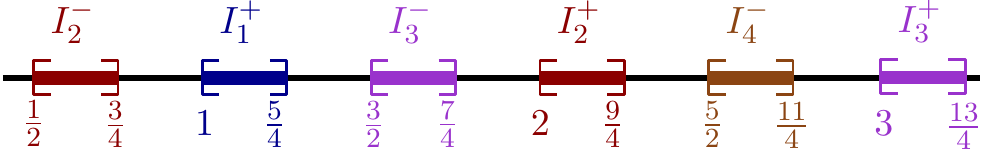}%
    \end{center}
    \vspace{-0.7cm}%
    \caption{The intervals $I_j^-$, $I_j^+$ out of which we sample the levels.
    The fractions are in multiples of $n/r$.}
    \figlab{intervals}%
\end{figure*}

The strategy goes as follows. For each index $i=2,\ldots,r-1$, consider
the pair of levels $L_{k_i^-}$, $L_{k_i^+}$, which we denote shortly
and respectively as $L_i^-$, $L_i^+$, and approximate both of them
simultaneously, using the following refinement of the algorithm
summarized in \thmref{apxlev}, with the same parameter $t=cn/r$ for
all pairs, where $c\ll 1/4$ is a sufficiently small constant.  Project
$L_i^-$ and $L_i^+$ onto the $xy$-plane, and overlap the resulting
planar maps $M_i^-$, $M_i^+$ into a single map $M_i^*$. Each vertex of
$M_i^*$ is either the projection of a vertex of one of the levels
$L_i^-$, $L_i^+$, or a crossing point of a pair of projected edges,
one from each level.

We now apply the preceding analysis to $M_i^*$, and
get a triangulation $T_i$ of the $xy$-plane, whose combinatorial
complexity is $O(|M_i^*|/t^2)$. We lift its vertices up to
both levels $L_i^-$, $L_i^+$, resulting in a corresponding pair of
triangulated terrains $T_i^-$, $T_i^+$, with identical
$xy$-projections.  We claim that $T_i^-$ approximates $L_i^-$ and
$T_i^+$ approximates $L_i^+$. Indeed, each boundary edge of a piece of
the $t^2$-division of $M_i^*$ is a portion of a projected edge of
either $L_i^-$ or $L_i^+$. Hence, traversing any portion of the
boundary of a piece, we encounter at most $O(t)$ edges of each of the
levels $L_i^-$, $L_i^+$. Specifically, when we go along an edge
  of $M_i^*$ that is a portion of some edge of $M_i^-$, say, we
  do not cross any edge of $M_i^+$, so when we lift the edge to
  the ``wrong'' terrain ($L_i^+$ in this case), we get two points
  that lie on the same face, and the property follows.
The arguments used above imply that, when
lifted to either of the two levels, the crossing number of the
corresponding path is at most $O(t)$, from which the claim follows.
Each triangle $\Delta$ of $T_i$ is lifted to a pair of triangles
$\Delta^-\in T_i^-$, $\Delta^+\in T_i^+$, and we connect them by a
vertical triangular prism $\Delta^*$ that has them as its bases. These
prisms are pairwise openly disjoint, and their union is the layer
$\Lambda_i$ between $T_i^-$ and $T_i^+$. Let $\Xi_i$ denote the
collection of these prisms.

For $i=1$  we
 project $L_1^+$, and
apply the preceding analysis to the resulting planar maps
$M_1^+$,  and get a triangulation $T_1$  of size
$O(|M_1|/t^2)$.
We lift $T_1$ up to $L_1^+$
and extend each triangle $\Delta$ of the lifted $T_1$ into a
 a semi-unbounded
prism, $\Delta^*$, that contains all the points vertically below $\Delta$.
We denote by $\Xi_1$ this
collection of semi-unbounded prisms.
We process $L_r^-$ analogously and obtain a collection $\Xi_r$ of semi-unbounded
prisms that contain all points vertically above the resulting lifted triangulation
$T_r$.

The union
$\Xi = \bigcup_i \Xi_i$ of our collections of vertical prisms  is the entire
3-space. These prisms are not pairwise openly disjoint, but each point
in $\Re^3$ is contained in the interiors of at most two prisms.
Informally, the layers $\Lambda_1,\ldots,\Lambda_r$ overlap in pairs
(but no three layers overlap), so that each layer is fully
triangulated by pairwise openly disjoint vertical prisms.

\begin{lemma}
    The expected size of $\Xi$ is $O(r^3)$.
\end{lemma}

\begin{proof}
    The overall number of prisms in $\Xi$ is, by \thmref{apxlev},
    \begin{equation}%
        \eqlab{compmi}%
        |\Xi| = O\left( \frac{1}{t^2}
            \sum_{i=1}^r |M_i^*| \right) .
    \end{equation}
    We have $|M_i^*| = O(|L_i^-| + |L_i^+| + |X_i|)$, where
    $X_i$ is the set of pairs $(e,e')$ of edges, where $e$ is an edge
    of $L_i^-$, $e'$ is an edge of $L_i^+$, and the
    $xy$-projections of $e$ and $e'$ cross each other. For $i=1$ we formally define
$M_1^* = M_1^+$ so $|M_1^*| = O(|L_1^+|)$ and for $i=r$
we define $M_r^* = M_r^-$ so $|M_r^*| = O(|L_r^-|)$.

    Estimating $\sum_i (|L_{k_i^-}| + |L_{k_i^+}|)$ is easy. Each
    level of $\Arr(H)$ appears in this sum with probability at most
    $4r/n$ (note that some levels will never be chosen), so the
    expected value is at most proportional to $4r/n$ times the
    complexity of $\Arr(H)$, namely,\footnote{Each vertex of $\Arr(H)$
       appears in three consecutive levels, and each edge appears in
       two, so features of $\Arr(H)$ may be drawn more than once, but
       at most three times.}
    $O((r/n)\cdot n^3) = O(n^2r)$.

    To estimate the expected value of $\sum_i |X_i|$, we note that
    each pair $(e,e')$ that is counted in this sum belongs to the set,
    call it $X_0$, of pairs that are accounted for in the bound in
    \lemref{x:t:edges}, with $q=7n/(4r)$, but our pairs constitute
    only a small subset of $X_0$. Specifically, by our choice of random levels,
    the probability of a pair $(e,e')\in X_0$ to appear in one of the sets
    $X_i$ is at most proportional to
    \begin{align*}
    (4r/n)^2\cdot|X_0| = O\left( (r/n)^2\cdot n^3\cdot (n/r) \right) =
    O(n^2r) .
    \end{align*}
    Substituting the separate bounds obtained so far in
    \Eqref{compmi}, we get that the expected size of $\Xi$ satisfies
    \begin{align*}
        |\Xi| = O\left( \frac{1}{t^2} \sum_{i=1}^r |M_i^*| \right) =
        O\left( \frac{r^2}{n^2} \cdot n^2r \right) = O(r^3) .
    \end{align*}
\end{proof}%

By applying this construction with $r'=r/c$ instead of $r$, for some sufficiently large constant $c$, we can guarantee that each prism does not intersect more than $n/r$ planes, at the cost of increasing the number of prisms by a constant factor.
We therefore obtain the following result.

\begin{theorem}
    For a set $H$ of $n$ non-vertical planes in $\Re^3$ in general
    position, and a parameter $r <n$,
    one can construct a layered $(1/r)$-cutting of $\Arr(H)$ of size
    $O(r^3)$. Specifically, we cover space by a set $\Xi$ of $O(r^3)$ vertical
    triangular prisms, such that each point is covered at most twice,
    and each prism is crossed by $\leq n/r$ planes of $H$.
    The top triangles of the prisms form $r$ polyhedral terrains, each
    approximating a suitable level of $\Arr(H)$, and the bottom
    triangles of the prisms form $r$ other polyhedral terrains of a similar nature.
\end{theorem}

One can construct such a cutting efficiently by using a
relative
$\left(\frac{1}{r},\frac{1}{2} \right)$-approximation
of size $O(r\log r)$ as in \secref{sec:efficient}.
We construct the arrangement $\Arr(S)$ of such a sample $S$, apply the construction as described above to obtain
a $(1/r')$-cutting of $\Arr(S)$ for  $r'=r/c$ where $c$ is a suitable fixed constant,  and claim that this is, with  probability $1-1/r^{O(1)}$, also
$(1/r)$-cutting of $\Arr(H)$. It takes linear time to sample $S$ and then $O(|S|^3)=O(r^3\log^3r)$ time
to construct $\Arr(S)$ and apply to it the algorithm described above.

One can construct the conflict lists using a standard range reporting data structure.
We preprocess $H$ into a data structure of size $s\ge n$, in time
$O(s \cdot {\rm polylog}(n))$, so that, for each query segment $e$, we can report
all the $k_e$ planes of $H$ that $e$ crosses in time
$O(\frac{n}{s^{1/3}}{\rm polylog}(n)) + k_e)$; See \cite{m-rsehc-93,Chan2012optimal} for details.
We query the structure with the $O(r^3)$ edges of the prisms of the cutting, and assemble
from the outputs, in a straightforward manner, the conflict lists of the prisms. The
overall running time is $O((\frac{r^3 n}{s^{1/3}} + s){\rm polylog}(n)  + \sum_e k_e)$.
We have that $\sum_e k_e = O(r^3(n/r)) = O(nr^2)$.
We choose $s=\max\{r^{9/4}n^{3/4},n\}$; this makes the running time
$O(nr^2 {\rm polylog}(n))$, as is easily checked.

\subsection{Approximate halfspace range counting}
\seclab{approx:r:count}%

In its dual setting, the problem is: Let $H$ be a set of $n$
nonvertical planes in $\Re^3$ in general position, and let $\eps>0$ be
an error parameter.  We wish to preprocess $H$ into a data structure
that supports queries of the form: For a query point $q$, count the
number of planes lying below $q$, up to a multiplicative factor of
$1\pm\eps$.  That is, if $q$ lies at level $k$, the answer should be
between $(1-\eps)k$ and $(1+\eps)k$.

Let $m = O(1 /\eps^{4/3})$. We construct and store the first $m$ levels of $\Arr(H)$
explicitly, each level as its own terrain. Formally,
for $j=0,\ldots, m$, we set $k_j = j$, and compute level $j$, denoted by $T'_j$, and
 its $xy$-projection denoted by $T_j$. Next, for deeper levels, we use the approximate
level construction. We take the sequence of levels
$k_{m+i} := m (1+\eps)^i$, for $i=1,\ldots,m'$, where
\begin{math}
    m' = \ceil{ \log_{1+\eps} \frac{n}{m}}%
    \approx%
    \frac{1}{\eps}\log n.
\end{math}
For each $i=m+1, \ldots, m+m'$, approximate  level $L_{k_i}$
up to an additive error of $\eps k_i$, let $T_i$ denote the underlying
triangulation in the $xy$-plane of the projection of the approximation,
and let $T'_i$ denote the approximating terrain, namely, the appropriate
lifting of $T_i$.
By construction,
it is easily checked that the terrains $T'_i$ do not
cross one another, and are therefore stacked on top of one another.
To answer an approximate (dual) halfspace range counting
query with a point $q$, we simply need to find two consecutive
terrains $T'_i$, $T'_{i+1}$ between which $q$ lies, and return
$m(1+\eps)^{i-m}$, say, as the approximate count, when $i>m$, or $i$
itself otherwise. For this we also construct a point location data structure over the xy-projection of
each of the first $m$ levels and over $T_i$ for $i = 1, \ldots, m'$.

By \thmref{apxlev}, for $i = 1, \ldots, m'$, the complexity of
$T_{m+i}$ (and of $T'_{m+i})$ is
\begin{align*}
    |T_{m+i}| =
    O\pth{ \frac{n}{\eps^3 k_{m+i}} }%
    =%
    O\pth{ \frac{n}{\eps^3 m (1+\eps)^i} }.
\end{align*}
Summing these bounds over $i$, we get
\begin{align*}
    \sum_{i=1}^{m'} \cardin{T_{m+i}} =
    O\pth{ \frac{n}{\eps^3 m} } \sum_{i=1}^{m'} \frac{1}{(1+\eps)^i}
    = O\pth{ \frac{n}{\eps^4 m} }.%
\end{align*}
Storing the first $m$ levels of $\Arr(G)$ requires $O(n m^2)$ space
(this bounds their overall complexity), so
both bounds are $O( n/\eps^{8/3})$, for $m = O(1/\eps^{4/3})$.
This bounds the storage used by our data structure.
To construct the data structure we spend $O(n\log n + nm^2)$ time to construct the first
$m$ levels \cite{c-rshrr-00}, and $O(nm^2\log n)$ time for the point location data structures over
the projections of the first $m$ levels.
Then, by  \thmref{a:shallow:cutting}, we spend $O(n+ \eps^{-6}\sum r_i\log^3 r_i)$ time
to construct the approximate levels $T'_i$, for $i = 1, \ldots, m'$, where
$r_i = n/k_{m+i}$. (The additive factor of $n$ in \thmref{a:shallow:cutting} accounts for the time it takes to draw the appropriate sample of $\frac{br_i}{\epsilon^2}\log r_i$ planes. It appear once in the bound above since we can draw all samples simultaneously.)
Since
$r_i = n/k_{m+i}=\frac{n}{m(1+\eps)^i}$, it follows that $\sum r_i = O(\frac{n}{m\eps}) = O(n\eps^{1/3})$ so the total time we
spend for constructing the approximate levels $T'_i$, for $i = 1, \ldots, m'$ is
$O(\eps^{-5\frac{2}{3}}n\log^3 n)$.
This includes the time it takes to construct a point location data structure over $T_i$ and dominates the total preprocessing time.

To answer a query with some point $q$, we run a binary search over the
terrains $T'_i$, and locate the $xy$-projection of $q$ in the relevant
planar maps $T_i$, thereby determining whether $q$ lies above or below
$T'_i$. The total cost of a query is therefore
\begin{align*}
    O\pth{ \log\left(\frac{1}{\eps}\log n\right) \cdot \log n } .%
\end{align*}

Afshani and Chan \cite{ac-arcd-09} showed how to avoid the binary
search for finding the right level, using a data structure of Kaplan
\etal \cite{krs-rmqrp-11}. Afshani and Chan use this latter structure
to find a rough approximation to the level.  Specifically, they find
an estimate $\hat{\ell}$ to the level $k$ of $q$ such that the
probability that $\hat{\ell}$ is off by a factor of (at least) $b$
from $k$ is  $O(1/b)$.
Then, instead of doing a binary search, they linearly search for the
right level, starting from the level in the hierarchy closest to
$\hat{\ell}$. The expected number of searches that they perform is
then $O(1)$ and these searches take $O(\log(n/(\eps k))$
time since each is a point location query over an arrangement of size $O(n/(\eps^3 k))$.  We can apply the exact
same technique using our approximate levels instead of the more
complicated refined shallow cuttings used in \cite{ac-arcd-09}, and
then get the following.

\begin{theorem}
    Let $H$ be a set of $n$ nonvertical planes in $\Re^3$, and let
    $\eps>0$ be a prescribed parameter. Ome can then construct a data
    structure of size $O(n/\eps^{8/3})$, in near-linear expected time,
 and that can answer approximate level
    queries in $\Arr(H)$, up to a relative error of $\eps$, in
    $O(\log(n/(\eps k)))$ expected time, where $k$ is the exact level
    of the query.
\end{theorem}

\noindent%
\textbf{Acknowledgments.} %
We thank J\'a{nos} Pach for pointing out that a variant of
\thmref{crucial} is already known.


\hypersetup{allcolors=black}%

\newcommand{\etalchar}[1]{$^{#1}$}
 \providecommand{\CNFSTOC}{\CNFX{STOC}}  \providecommand{\CNFX}[1]{
  {\em{\textrm{(#1)}}}}  \providecommand{\CNFSoCG}{\CNFX{SoCG}}
  \providecommand{\CNFSODA}{\CNFX{SODA}}
  \providecommand{\CNFFOCS}{\CNFX{FOCS}}


\begin{figure*}[p]
    \centerline{\animategraphics%
       [autoplay,autoresume,width=0.9\linewidth,controls,loop,buttonsize=4em]%
       {2}{figs/animation_1}{0}{22}}
    \caption{Animation of algorithm -- you would need Acrobat reader
       to see the animation - click the figure to make it start.}
    \figlab{figanim}
\end{figure*}

\end{document}